\documentclass[11pt,reqno]{amsart}
\usepackage[english]{babel}
\usepackage{mathtools,amsmath,amssymb}
\usepackage{mathrsfs,dsfont}
\usepackage{enumitem}
\usepackage{soul,xcolor}
\usepackage{tikz-cd}
\usepackage{stmaryrd}
\usepackage{slashed}

\usepackage[pdfencoding=auto, psdextra]{hyperref}
\hypersetup{colorlinks=true,allcolors=[rgb]{0,0,0.6}}
\usepackage[hmargin=2.7cm,vmargin=2.6cm]{geometry}

\usepackage{dsfont}
\usepackage{amsthm}
\usepackage{xcolor}
\usepackage[normalem]{ulem}

\newcommand{\N}{\mathbb{N}}
\newcommand{\R}{\mathbb{R}}
\newcommand{\C}{\mathbb{C}}

\newcommand{\Id}{\mathds{1}}

\newcommand{\Hi}{\mathscr{H}}
\newcommand{\eps}{\varepsilon}
\newcommand{\Tr}{\mathrm{Tr}}
\newcommand\norm[1]{\left\lVert#1\right\rVert}
\newcommand\abs[1]{\left\lvert#1\right\rvert}

\newcommand{\di}{\mathrm{d}}

\numberwithin{equation}{section}

\theoremstyle{plain} 
\newtheorem{thm}[equation]{Theorem}
\newtheorem{cor}[equation]{Corollary}
\newtheorem{lem}[equation]{Lemma}
\newtheorem{prop}[equation]{Proposition}

\theoremstyle{definition}
\newtheorem{defn}[equation]{Definition}

\theoremstyle{remark}
\newtheorem{rem}[equation]{Remark}

\allowdisplaybreaks

\begin{document}

\title[Derivation of Landau-Pekar]{Wigner measure approach to the derivation of the Landau-Pekar equations in the mean-field limit}

\author[R.\ Gautier]{Raphaël Gautier} 
\address{Dipartimento di
  Matematica\\Politecnico di Milano\\P.zza Leonardo da Vinci 32\\20133,
  Milano\\Italy\\\& Université de Rennes\\ Campus de Beaulieu\\263, avenue du
  Général Leclerc\\35042 RENNES CEDEX\\France}
\email{raphael.gautier@univ-rennes.fr}

\date{}

\subjclass[2020]{}

\begin{abstract}

We provide a rigorous derivation of the Landau-Pekar equations from the Fröhlich Hamiltonian in the mean-field limit using Wigner measure techniques. On the classical side, we extend the global well-posedness results up to $L^2 \oplus L^2$. For obtaining the classical limit, we make a crucial use on the quantum side of the Gross transform, allowing us to study the system on the energy space, and with no ultraviolet cutoff. 

\end{abstract}

\maketitle
\tableofcontents

\maketitle

\section{Introduction} \label{Introduction}
The Fröhlich Hamiltonian describes the energy of a quantum system consisting of an electron coupled with a phonon field. The electron, together with the induced polarization of the field, can be viewed as a (quasi)particle, called the polaron. In this work, we consider this Hamiltonian for a system of $n$ bosonic particles coupled with the phonon field, which is represented by a semiclassical bosonic field, in the mean-field limit, i.e., in the limit $n \eps \rightarrow 1$. We prove, in this regime, Bohr's correspondence principle, linking the quantum evolution by the Fröhlich Hamiltonian to the classical evolution described by the flow of the corresponding classical system, the Landau-Pekar equations. 
\newline \newline
Fröhlich's polaron has been extensively studied in the so-called strong coupling limit, where a lot of investment has been put on the study of the self-energy or the effective mass of the quasi-particle, see for instance \cite{lieb1958,lieb1997,Sei02,Mol06,MS07,Ric14,Sei19,leopoldmitrouskas2021, BS23, BS24, BM24} on the analytic side via effective evolutions of the system, and \cite{DV83,DS20,BP23,Sell24,BSS25} for a probabilistic side based on point process representation of the path measure of the polaron.
Concerning the classical system, that is the Landau-Pekar equations, one can check \cite{LR13a,LR13b,FG17,Gri17,GSS17,leopoldmitrouskas2021,leopold2021,leopoldrademacher2021,FRS22}, and for quantum corrections, see e.g. \cite{FS19,FS21}.
\newline \newline 
In this work, we take a measure-theoretical approach using the Wigner, or semiclassical measures.
This allows to take initial quantum states lying in the energy space $H^1 \oplus L^2$, with no additional regularity restriction.
The cost of these techniques however is the inability of describing the rate of convergence. This compromise can be seen as a complement to the work \cite{leopold2021}.
By construction, this paper, in a mathematical aspect, resembles more the work \cite{ammari2017sima}, which uses similar techniques to handle the renormalized Nelson model in the same regime.
\newline \newline  
For studying the mean-field limit, we follow the approach initiated in \cite{ammari2008ahp} on the semiclassical limit of second quantized systems arising from an infinite dimensional phase-space. Along the same lines has followed multiple work on Egorov-type theorems in many nonrelativistic quantum field models.
Let us mention the Nelson model with cutoff \cite{ammari2014jsp,farhat2024}, or without \cite{ammari2017sima}, and the Pauli-Fierz model \cite{afh22}.
Another point of view on systems that are partially semiclassical led to the works \cite{CF18,correggi2021jems, correggi2020arxiv} studying Egorov type theorems and ground state energies for the Nelson, Pauli-Fierz and Fröhlich models in the latter scaling, and \cite{fantechi2024arxiv}, in the context of decoherence for a spin system coupled to a reservoir. 
\\ \\
The interaction in the electron-phonon system is singular enough to complicate the domain of definition of the Hamiltonian, making it unable to use a Rellich-Kato argument. For that purpose, the analysis of the polaron Hamiltonian usually comes along with the analysis of the Gross transform, introduced in \cite{gross1692}, transforming the latter into a dressed Hamiltonian whose domain is the domain of the free Hamiltonian, that is, the energy space. A singularity of this kind also appears in the Nelson model \cite{griesemer2018}, and we will make a crucial use of the results of \cite{griesemer2016}. Our strategy will then be to study the mean-field dynamics of the dressed model and going back to the original model by inverting the Gross transform.
\\ \\ 
\textbf{Main results.} The Fröhlich Hamiltonian for a system of $n$ bosonic particles interacting with a bosonic field is
\begin{equation} \label{quantumhamiltonian}
H_\infty^{(n)} = \sum_{j=1}^n - \eps \Delta_{x_j} \otimes \Id +\Id \otimes  \di \Gamma_\eps (\Id) + \eps \sum_{j=1}^n \left(a_\eps \left( \frac{e^{-i k .x_j}}{\abs{k}^{(d-1)/2}} \right)  +  a_\eps^* \left( \frac{e^{-i k .x_j}}{\abs{k}^{(d-1)/2}} \right)\right)
\end{equation}
which is acting on the Hilbert space
$$
\Hi_n := (L^2(\R^d))^{\otimes_s n} \otimes \mathscr F_s(L^2(\R^d)),
$$
with $d \geq 3$. We aim to study the mean-field limit of the evolution of the system, that is, in the limit $n \eps \rightarrow 1$. The classical system is described by the classical Hamiltonian
\begin{equation} \label{classicalhamiltonian}
    h_\infty(u,\alpha) = \norm{\nabla u}_{L^2}^2 + \norm{\alpha}_{L^2}^2 + \int A_\alpha (x) \abs{u(x)}^2 \di x,
\end{equation}
where 
$$
A_\alpha(x) = 2 \mathrm{Re} \mathcal F \left(  \frac{\overline{\alpha}}{\abs{k}^{(d-1)/2}} \right)(x) = \int( \overline{\alpha}(k) e^{-i x .k} + \alpha(k) e^{i x.k} ) \frac{\di k}{\abs{k}^{(d-1)/2}}.
$$
 The corresponding Hamilton-Jacobi equations, which are called Landau-Pekar equations, are
 \begin{equation} \label{PDE}
\begin{cases}
    i \partial_t u = - \Delta u + A_\alpha u, \\
    i \partial_t \alpha = \alpha + \frac{1}{\abs{k}^{(d-1)/2}} \mathcal F (\abs{u}^2 ).    
\end{cases}
\end{equation}

We have the following well-posedness results:

\begin{prop}
    \label{thmgwp}
 (global existence in $H^1 \oplus L^2$)
For any $(u_0,\alpha_0) \in H^1 \oplus L^2$, there exists a unique $(u,\alpha) \in \mathscr C (\R,H^1 \oplus L^2)$ solution of \eqref{PDE} with $(u(0),\alpha(0)) = (u_0,\alpha_0)$. Furthermore, the mass $\norm{u}_{L^2}^2$ and the energy $h_\infty$ are conserved by the evolution.
\end{prop}

\begin{thm} \label{thmgwpl2}
 (global existence in $L^2 \oplus L^2$)
For any $(u_0,\alpha_0) \in L^2 \oplus L^2$, there exists a unique $(u,\alpha) \in \mathscr C (\R,L^2) \cap L^2_{\mathrm{loc}}(\R,L^{2d/(d-2)}) \times \mathscr C(\R,L^2)$ solution of \eqref{PDE} with $(u(0),\alpha(0)) = (u_0,\alpha_0)$. Furthermore, the mass $\norm{u}_{L^2}^2$ is conserved by the evolution.
\end{thm}

Proposition \ref{thmgwp} is well known and we provide in Section \ref{Study of the classical model} the proof for completeness. From a mathematician's perspective, we also consider the case with initial data in the phase-space $L^2 \oplus L^2$, and extend the global well-posedness results using (half) dispersive properties of the Landau-Pekar equations. In either case we call $\phi_t$ the Hamiltonian flow of \eqref{PDE}. We aim to link the evolution of the quantum system and its classical counterpart. For that purpose, it will be convenient to rewrite the Hamiltonian $H_\infty^{(n)}$ in a fully second quantized form as

$$
H_\infty = \int_{\R^d} \psi_x^* (- \Delta_x) \psi_x \di x +\int_{\R^d} a_k^* a_k \di k + \iint_{\R^d \oplus \R^d} \psi_x^* \left(a_k^*\frac{e^{-i k .x}}{\abs{k}^{(d-1)/2}} + a_k \frac{e^{i k .x}}{\abs{k}^{(d-1)/2}}  \right)\psi_x \di k \di x,
$$
acting on 
$$
\Hi := \mathscr F_s (L^2 \oplus L^2) \cong \mathscr F_s (L^2) \otimes \mathscr F_s(L^2),
$$
in a way that 
$$
H_\infty|_{\Hi_n} = H_\infty^{(n)}.
$$
We consider initial normal states (positive, trace-class and of unit trace operators) satisfying the following assumptions:
\begin{align}
    &\exists C>0, \forall \eps \in (0,1], \forall k \in \N, \Tr[\rho_\eps N_1^k] \leq C^k. \label{assumption1} \\
    &\exists C>0, \forall \eps \in (0,1], \Tr[\rho_\eps H_0] \leq C. \label{assumption2}
\end{align}

The main result of this paper is the following: 

\begin{thm} \label{thmdynamics}
    Consider a family of normal states $(\rho_\eps)_{\eps \in (0,1]}$ satisfying Assumptions \eqref{assumption1} and \eqref{assumption2}. Assume that
    $$
    \mathscr M (\rho_\eps , \eps \in (0,1]) = \{ \mu_0 \} \subset \mathcal P (L^2 \oplus L^2).
    $$
    Then, $\forall t \in \R$,
    $$
    \mathscr M (e^{- \frac{i t}{\eps} H_\infty} \rho_\eps e^{\frac{i t}{\eps} H_\infty} , \eps \in (0,1]) = \{(\phi_t)_\# \mu_0 \} \subset \mathcal P (L^2 \oplus L^2),
    $$
    where $\phi_t$ is the classical Hamiltonian flow of the undressed evolution.
\end{thm}
Here $\mathscr M(\rho_\eps, \eps \in (0,1])$ denotes the set of Wigner measures of the family $(\rho_\eps)_{\eps \in (0,1]}$, see Section \ref{Derivation of the classical system} for more details. The Fröhlich Hamiltonian is singular in the sense that the form factor $\frac{e^{-i k .x_j}}{\abs{k}^{(d-1)/2}}$ is not square integrable. A straightforward KLMN argument shows that $H_\infty^{(n)}$ is still well defined as a quadratic form on the form domain of the free Hamiltonian 
$$
H_0^{(n)} = \sum_{j=1}^n - \eps \Delta_{x_j} \otimes \Id +\Id \otimes  \di \Gamma_\eps (\Id).
$$
However, such an argument cannot work for the domain of the operators, and in fact, it was proven in \cite{griesemer2016} that $D(H_0^{(n)}) \cap D(H_\infty^{(n)}) = \emptyset$. They obtained in addition a characterization of the domain by means of the Gross transform $
U_\infty = e^{\frac{i}{\eps} T_\infty}$, where 
$$
T_\infty = \iint \psi_x^* \left( - i \frac{\Id_{\abs{k}\geq \sigma_0}}{(1+\abs{k}^2) \abs{k}^{(d-1)/2}} a_k^* e^{-ik .x} + h.c. \right) \psi_x \di k \di x,
$$
with $\sigma_0$ a fixed infrared cutoff. The dressed Hamiltonian
$$
\hat H_\infty = U_\infty H_\infty U_\infty^*,
$$
is self-adjoint on $D(H_0)$ and no singularity appears. On the classical side, this translates to the evolution not being governed by $h_\infty$ but by the dressed Hamiltonian
\begin{align*}
    \hat h_\infty (u,\alpha) &= \norm{\nabla u}_{L^2}^2 + \norm{\alpha}_{L^2}^2 +  \int 2 \mathrm{Re} \langle \alpha, f_{\sigma_0} e^{-ik.x} \rangle \abs{u(x)}^2  \di x \\
        &+ \iint \abs{u(x)}^2 V_\infty(x-y) \abs{u(y)}^2 \di x \di y \\
        &+ \int (2 \mathrm{Re} \langle \alpha, k B_\infty e^{-ik.x} \rangle)^2 \abs{u(x)}^2 \di x  \\
        & - 2 \int   \overline{u(x)} ( \langle \alpha, k B_\infty e^{-i k . x} \rangle  .D_{x} + D_x . \langle k B_\infty e^{-i k .x},\alpha \rangle) u(x) \di x,
    \end{align*}
where $D_x = -i \nabla_x$.
\newline
\newline
\textbf{Structure of the paper.} 
In the rest of Section \ref{Introduction} we give mathematical background on bosonic Fock spaces and operators arising there. We devote Section \ref{Definition of the Hamiltonian} to the definition of the quantum Hamiltonian in the second quantization framework, to precise how the removal of the cutoff is handled. In Section \ref{Study of the classical model} we analyse the classical Landau-Pekar equations \eqref{PDE}, meaning global well-posedness, conservation laws, as well as classical study of the Gross dressing transform, and a link between the undressed and dressed evolutions. Finally, in Section \ref{Derivation of the classical system} we study the quantum evolution via a Duhamel formula for the noncommutative Fourier transform of the evolved dressed state. The mean-field limit thus yields a transport equation that is solved by general measure-theoretical tools. The last step is then to undo the Gross tranform to describe the quantum evolution. \\

\textbf{Second Quantization framework.} Let us recall here the standard background of second quantization, with a scaling that is suitable for our semiclassical analysis. An extended overview can be found in \cite{araiasao}. 
\begin{defn}
    The bosonic Fock space over a separable Hilbert space $\mathscr Z$, considered as the classical phase-space, is the Hilbert space 
    $$
    \mathscr F_s(\mathscr Z) := \bigoplus_{n \in \N} \mathscr Z^{\otimes_s n},
    $$
    where $\mathscr Z^{\otimes_s n}$ denotes the $n$-fold symmetric tensor product of $\mathscr Z$. Other standard notations are $\mathscr F_b(\mathscr Z), \Gamma_b(\mathscr Z)$ or $ \Gamma_s(\mathscr Z)$.
\end{defn}

If one thinks of $\mathscr Z$ as describing the configuration space of a single particle, then $\mathscr Z^{\otimes_s n}$ describes a system of $n$ such indistinguishable particles. In many quantum field models however, the number of particles may not be conserved, and $\mathscr F_s (\mathscr Z)$ allows to describe a system with a possibly varying number of particles. The vector 
$$
\Omega = 1 \oplus 0 \oplus 0 \oplus...
$$
is called the vacuum and describes a state with no particles. To any self-adjoint operator $A$ acting on the one particle Hilbert space $\mathscr Z$ corresponds a self-adjoint operator $\di \Gamma (A)$ in the following way:
$$
\di \Gamma(A)_{|\mathscr Z^{\otimes_s n}} := \eps \sum_{j=1}^n \Id \otimes ... \otimes A \otimes ... \otimes \Id.
$$
In other words, $A$ is acting on every particle individually. It is essentially self-adjoint on the finite number of particle subspace
$$
\mathscr F_{\mathrm{fin}} (\mathscr Z) = \bigoplus_{n\in \N}^{\mathrm{alg}} \mathscr Z^{\otimes_s n}.
$$
Here the subscript ``$\mathrm{alg}$" is referring to the algebraic direct sum, i.e. there is only a finite number of nonzero components. 
In this paper, we will consider the phase-space to be $\mathscr Z = L^2(\R^d) \oplus L^2(\R^d)$, in which case we denote the number of particles of each side by
$$
N_1 = \di \Gamma(\Id) \otimes \Id, \ \ \ N_2 = \Id \otimes \di \Gamma(\Id).
$$
Let us now define an important class of operators:

\begin{defn}
    Let $f \in \mathscr Z$. The formula
    $$
      (a^*(f) \varphi)_n :=  \sqrt{\eps(n+1)} S_{n+1}(f \otimes \varphi_n) , \varphi =  (\varphi_n)_n \in \mathscr F_s(\mathscr Z),
    $$
    defines a closable operator $a^*(f)$ on $\mathscr F_s (\mathscr Z)$. His closure, still denoted $a^*(f)$, is called creation operator. Here $S_n$ denotes the orthogonal projection of $\mathscr Z^{\otimes n}$ to $\mathscr Z^{\otimes_s n}$. Its adjoint is denoted
    $$
    a(f) = (a^*(f))^*
    $$
    and is called annihilation operator.
\end{defn}

\begin{rem}
    $a^*(f)$ sends $\mathscr Z^{\otimes_s n}$ to $\mathscr Z^{\otimes_s (n+1)}$, ``creating" a particle, and $a(f)$ sends $\mathscr Z^{\otimes_s (n+1)}$ to $\mathscr Z^{\otimes_s n}$, ``destroying" a particle.
\end{rem}

\begin{prop}(canonical commutation relations) The following holds on $\mathscr F_{\mathrm{fin}} (\mathscr Z)$: for all $f,g \in \mathscr Z$,
$$
\begin{cases}
    [a(f),a(g)] = [a^*(f),a^*(g)] = 0, \\
    [a(f),a^*(g)] = \eps \langle f,g\rangle_{\mathscr Z}.
\end{cases}
$$
\end{prop}
In our case $\mathscr Z = L^2(\R^d) \oplus L^2(\R^d)$, we will denote by the letter $\psi$ the creation and annihilation operators acting on the first $L^2(\R^d)$, and by $a$ those acting on the second $L^2(\R^d)$. In $\mathscr Z$ is of the form $\mathscr Z = L^2 (X,\mathscr A,\mu)$, one often finds in the literature the notation $a_x, a_x^*, x \in X$. Here, $a_x, a^*_x$ are seen as operator valued distributions, satisfying
$$
a(f) = \int_X a_x \overline{f(x)} \mu (\di x) , \ \ \ a^*(f) = \int_X a^*_x f(x) \mu(\di x),  
$$
and the commutation relations become
$$
\begin{cases}
    [a_x,a_y] = [a^*_x,a^*_y] = 0, \\
    [a_x,a^*_y] = \eps \delta(x-y).
\end{cases}
$$
With these notations, it is possible to write second quantized operators as
$$
\di \Gamma (A) = \int a_x^* A_x a_x \di x.
$$
For a composite Hilbert space of the form $\mathscr K \otimes \mathscr F_s(\mathscr Z)$, it is also possible to define for any $F \in \mathscr B(\mathscr K,\mathscr K \otimes \mathscr Z)$, a generalized creation operator $a^*(F)$ via the formula
$$
a^*(F) (\varphi \otimes \phi_n) = \sqrt{\eps(n+1)} S_{n+1}(F \varphi \otimes \phi_n) , \varphi \in \mathscr K, \phi_n \in \mathscr Z^{\otimes_s n}. 
$$
Here $F \varphi \otimes \phi_n$ is seen as an element of $\mathscr K \otimes \mathscr Z^{\otimes (n+1)}$ and $S_{n+1}$ is the orthogonal projection on $\mathscr K \otimes \mathscr Z^{\otimes_s (n+1)}$. When taking $F_f : \varphi \mapsto \varphi \otimes f$, we recover the usual creation and annihilation operator. Indeed, a straightforward computation shows that in this case,
$$
a^*(F_f) = \Id_{\mathscr K} \otimes a^*(f).
$$
In our case we will encounter for instance the space $L^2(\R^d) \otimes \mathscr F_s(L^2(\R^d))$ and generalized creation and annihilation operators will often appear with
$$
G_g : \varphi \in L^2 \longmapsto \left( (x,k) \mapsto g(k) e^{-i k . x} \varphi(x) \right),
$$
for a suitable measurable map $g : \R^d \rightarrow \R$. We will in particular use the notation $a^\# (g e^{-i x . k}) := a^\# (G_g)$.
\begin{defn}
    Let $f \in \mathscr Z$. We define the field operator
    $$
    \phi(f) = \frac{1}{\sqrt{2}}( a(f) + a^*(f)),
    $$
    and the Weyl operator 
    $$
    W(f) = e^{i \phi(f)}.
    $$
\end{defn}
The Weyl operator can be interpreted as the quantum, or noncommutative version of the Fourier character $z \mapsto e^{2 i \pi \mathrm{Re} \langle z,f \rangle}$ and will be use as an elementary brick for the construction of Wigner measures. They satisfy the Weyl commutation relations:
$$
W(f) W(g) = W(f+g) e^{-\frac{i \eps}{2} \mathrm{Im} \langle f,g \rangle}.
$$
We will sometimes denote the operators $a^\#(f), \di \Gamma(A), W(f)$ as $a^\#_\eps(f), \di \Gamma_\eps (A), W_\eps(f)$ to enlight their dependance with the semiclassical parameter.

\section{Definition of the Hamiltonian} \label{Definition of the Hamiltonian}
For $n \in \N,0 \leq \sigma <+\infty$, we consider the polaron Hamiltonian
$$
H_\sigma^{(n)} = \sum_{j=1}^n - \eps \Delta_{x_j} \otimes \Id +\Id \otimes  \di \Gamma_\eps (\Id) + \eps \sum_{j=1}^n (a_\eps (  f_\sigma(k) e^{-i k .x_j})+h.c.)
$$
acting on $\Hi_n = (L^2(\R^d))^{\otimes_s n} \otimes \mathscr F_s (L^2(\R^d))$, in dimension $d \geq 3$. Here, $$f_\sigma(k) = \frac{\Id_{\abs{k}\leq \sigma}}{\abs{k}^{(d-1)/2}}, $$ is called the form factor. We call $H_0^{(n)} = \sum_{j=1}^n - \eps \Delta_{x_j} \otimes \Id +\Id \otimes  \di \Gamma_\eps (\Id)$ the free Hamiltonian and $H_{\sigma,I}^{(n)} =  \eps \sum_{j=1}^n (a (  f_\sigma e^{-i k .x_j})+h.c.)$ the interaction. They are scaled in a mean-field way, every term being of order $\eps n \sim 1$ on the left side. We plan to study the dynamics generated by this Hamiltonian, in the limit $n \eps \sim 1$. For convenience, we put our Hamiltonian in the framework of second quantization.
We define, 
$$
H_\sigma = \int_{\R^d} \psi_x^* (- \Delta_x) \psi_x \di x +\int_{\R^d} a_k^* a_k \di k + \iint_{\R^d \oplus \R^d} \psi_x^* (a_k^* f_\sigma (k)e^{-i k .x} + h.c.)\psi_x \di k \di x,
$$
acting on $\Hi:= \mathscr F_s(L^2) \otimes \mathscr F_s (L^2) \cong \mathscr F_s(L^2 \oplus L^2)$. Here $\psi_x, \psi_x^*$ denote the annihilation and creation operators over the first phase-space $L^2$, and $a_k,a_k^*$ over the second, so that when restricting to $\Hi_n := (L^2)^{\otimes_s n} \otimes \mathscr F(L^2)$, we obtain
$$
H_\sigma|_{\Hi_n} = H_\sigma^{(n)}.
$$
Since $f_\infty$ is not an $L^2$ function, it is not clear however if $H_\sigma$ makes sense on $\Hi$ as a self-adjoint operator in the limit $\sigma \longrightarrow +\infty$. A way to define it is using the following. 
\begin{defn}
    Let $0 <\sigma_0 \leq \sigma \leq \infty.$ Consider
    $$
    T_\sigma = \iint \psi_x^* (i B_\sigma(k) a_k^* e^{-ik .x} + h.c.) \psi_x \di k \di x,
    $$
    with $$
    B_\sigma (k) = -\frac{\Id_{\sigma_0 \leq \abs{k}\leq \sigma}}{(1+\abs{k}^2)\abs{k}^{(d-1)/2}}.$$
    Then the Gross transform is defined as 
    $$
    U_\sigma = e^{\frac{i}{\eps}T_\sigma}.
    $$
\end{defn}

\begin{rem}
    $B_\sigma \in L^2$ for any $\sigma_0\leq \sigma \leq \infty$. Thus the above expression makes sense when $\sigma = \infty$, allowing to study the dynamics of $T_\infty$, and thus describing well $U_\infty$ itself. 
\end{rem}

\begin{rem}
    $T_\sigma$ leaves $\Hi_n$ invariant, so that one can compute
    $$
    U_\sigma|_{\Hi_n} = e^{\frac{i}{\eps} T_\sigma|_{\Hi_n}} = e^{i \sum_{j=1}^n (a^*(i B_\sigma e^{-ik .x_j}) + a(i B_\sigma e^{-ik .x_j})) \ }.
    $$
\end{rem}

\begin{lem} \cite[Theorem 3.7]{griesemer2016} \label{lemmaappendix1} For any $n \in \N$, $\sigma_0\leq \sigma <+\infty$, let 
$$
\hat H_\sigma^{(n)} := U_\sigma^{(n)} H_\sigma^{(n)} U_\sigma^{*(n)}, 
$$
with domain $D(\hat H_\sigma^{(n)}) = D(H_0^{(n)})$. Then,
\begin{align*}
\hat H_\sigma^{(n)} &=  H_0^{(n)} + \hat H_{\sigma,I}^{(n)} \\&= H_{\sigma_0}^{(n)} + \eps^2 \sum_{1\leq i \neq j\leq n } V_\sigma (x_i - x_j)  \\ &+ \eps \sum_{j=1}^n [a^*(k B_\sigma e^{-ik.x_j})^2+ a(k B_\sigma e^{-ik.x_j})^2 + 2 a^*(k B_\sigma e^{-ik.x_j}) a(k B_\sigma e^{-ik.x_j}) \\ &-2D_{x_j} . a(k B_\sigma e^{-ik.x_j})-2 a^*(k B_\sigma e^{-ik.x_j}).D_{x_j}]\\ &+\eps^2 n \langle B_\sigma, B_\sigma + 2 f_\sigma \rangle,
\end{align*}
where
$$
V_\sigma (x) = \mathrm{Re} \int_{\R^d} (\abs{B_\sigma (k)}^2+2 B_\sigma(k) f_\sigma(k)) e^{-i k . x} \di k.
$$
\end{lem}

\begin{rem}
    By opposition to the Nelson model, the constant term $\langle B_\sigma, B_\sigma + 2 f_\sigma \rangle$ does not diverge when removing the cutoff, i.e. in the limit $\sigma \rightarrow \infty$, so we have no reason of substracting it. In other words, the Gross transform is used only for the definition of the Hamiltonian, but no renormalization is needed. In fact, this term, which, in a second quantization framework, can be written as $\eps  \langle B_\sigma, B_\sigma +2  f_\sigma \rangle \int \psi_x^* \psi_x \di x,$ will vanish at the classical level due to the extra $\eps$, see Lemma \ref{N_1wigner} below.
\end{rem}

We have the following KLMN bound:
\begin{lem} \label{klmnbound}
    Let $n \in \N$. There exists $\sigma_0(n \eps )>0$, $a<1$ and $C_{a, n \eps} \in \R$ such that for all $\sigma \in [\sigma_0,+\infty],$ the following holds:
    $$
    \forall \varphi_n \in D(H_0^{(n)}), \abs{ \langle \varphi_n , \hat H_{\sigma,I}^{(n)} \varphi_n \rangle} \leq a \langle \varphi_n , H_0^{(n)} \varphi_n \rangle + C_{a,n \eps} \norm{\varphi_n}^2_{\Hi_n}.
    $$
\end{lem}

\begin{proof}
    Let $\varphi_n \in D(H_0^{(n)})$. First, we have, for all $a<1$, and $\# \in \{ \emptyset , * \}$,
    \begin{align*}
        \abs{ \langle \varphi_n, \eps \sum_{j=1}^n a^\#(f_{\sigma_0} e^{-i k . x_j})  \varphi_n \rangle } &= \abs{ \eps n \langle \varphi_n , a^\# (f_{\sigma_0} e^{-i k . x_1}) \varphi_n \rangle } \\
        & \leq \eps n \norm{\varphi_n} \norm{a^\# (f_{\sigma_0} e^{-i k . x_1}) \varphi_n} \\
        & \leq \eps n \norm{f_{\sigma_0}}  \norm{\varphi_n} \norm{\sqrt{N_2+1} \varphi_n} \\
        & \leq a \langle \varphi_n, N_2 \varphi_n \rangle + C_{a, n \eps,\sigma_0} \norm{\varphi_n}^2.
    \end{align*}
    Similarly,
    \begin{align*}
        \abs{ \langle \varphi_n, \eps \sum_{j=1}^n a^\#(k B_{\sigma} e^{-i k . x_j}) a^\#(k B_{\sigma} e^{-i k . x_j})  \varphi_n \rangle } 
        & \leq \eps n \left( \norm{k B_{\sigma}}  \norm{\varphi_n} \norm{\sqrt{N_2+1} \varphi_n} \right)^2 \\
        & \leq \eps n \norm{k B_{\infty}}^2 (a \langle \varphi_n, N_2 \varphi_n \rangle  + C_a \norm{\varphi_n}^2) \\
        & \leq  \eps n \norm{k B_{\infty}}^2  (a \langle \varphi_n , H_0^{(n)} \varphi_n \rangle + C_a \norm{\varphi_n}^2_{\Hi_n}),
    \end{align*}
    \begin{align*}
        \abs{ \langle \varphi_n, \eps \sum_{j=1}^n D_{x_j} . a^\#(k B_{\sigma} e^{-i k . x_j})  \varphi_n \rangle } & \leq \eps \norm{k B_\sigma} \sum_{j=1}^n \norm{D_{x_j} \varphi_n} \norm{ \sqrt{N+1} \varphi_n } \\
        & \leq \eps \norm{ k B_\sigma} \sum_{j=1}^n \langle \varphi_n, (-\Delta_{x_j}+N_2+1) \varphi_n \rangle  \\
        &\leq \eps n  \norm{k B_\infty} ( \langle \varphi_n, H_0^{(n)} \varphi_n \rangle + \norm{\varphi_n }^2 ).
    \end{align*}
    The last term is easy to deal with, as $V_\sigma$ is uniformly bounded in $L^2$:
    $$
        \abs{ \langle \varphi_n, \eps^2 \sum_{1 \leq i < j \leq n} V_{\sigma}(x_i - x_j)  \varphi_n \rangle } \leq (\eps n)^2 \underset{\sigma \geq \sigma_0}  \sup \norm{V_\sigma}_{\infty}  \norm{\varphi_n}^2.
    $$
    Hence, it suffices to choose $\sigma_0 = \sigma_0 ( n\eps)$ such that      $ \norm{k B_{\infty}}^2 \leq (\eps n)^{-1}$ and $ \norm{k B_\infty}  \leq a (\eps n)^{-1}$ to conclude, the bound on the constant term $\eps^2 n \langle B_\sigma, B_\sigma + 2 f_\sigma \rangle$ being immediate.
\end{proof}

We are now in position to define the dressed Hamiltonian $\hat H_\infty$ in a suitable mean-field regime, and without ultraviolet cutoff.

\begin{defn} ($\hat H_{\infty}$) We define $\hat H_\infty$ by its restrictions
$$
\hat H_\infty|_{\Hi_n} = \begin{cases}
    \hat H^{(n)}_{\infty} &\text{if } n \eps \leq C,\\
    0 &\text{if else},
\end{cases}
$$
where $C$ is the constant appearing in Assumption \eqref{assumption1}, and with domain
$$
D(\hat H_\infty) = \{ \varphi = (\varphi_n)_n \in \Hi = \bigoplus_{n \in \N} \Hi_n | \forall n \leq \lfloor C \eps^{-1} \rfloor, \varphi_n \in D(H_0^{(n)})  \}.
$$
\end{defn}
By undoing the dressing we are now able to give a definition of the undressed Hamiltonian.
\begin{defn} ($ H_{\mathrm{\infty}}$)
    We define
    $$
    H_\infty = U_\infty^* \hat H_\infty U_\infty,
    $$
    with domain
    $$
    D(H_\infty) = \{ \varphi = (\varphi_n)_n \in \Hi = \bigoplus_{n \in \N} \Hi_n | \forall n \leq \lfloor C \eps^{-1} \rfloor, \varphi_n \in U_\infty^{*(n)} D(H_0^{(n)})  \}.
    $$
\end{defn}

\section{Study of the classical model}  \label{Study of the classical model}
We study the well posedness of the Landau-Pekar equations,
\begin{equation*}
\begin{cases}
    i \partial_t u = - \Delta u + A_\alpha u, \\
    i \partial_t \alpha = \alpha + \frac{1}{\abs{k}^{(d-1)/2}} \mathcal F (\abs{u}^2 ),    
\end{cases}
\end{equation*}
where we recall that
$$
A_\alpha (x) =  2 \mathrm{Re} \mathcal F \left( \frac{\overline \alpha}{\abs{k}^{(d-1)/2}} \right).
$$

\begin{rem}
    The equation for $u$ can be interpreted as a Schrödinger equation with varying potential $A_\alpha$. This observation will be crucial in the sense that we will use standard dispersive estimates and Strichartz estimates for proving the global well-posedness in the physical space $L^2 \oplus L^2$. For that purpose we recall that a pair $(p,q)$ is called admissible whenever $p,q \in [2,+\infty]$, and
$$
\frac{2}{p}+\frac{d}{q} = \frac{d}{2}.
$$
We notice that $(2,2d/(d-2))$ is an admissible pair (for $d \neq 2$, which is excluded in this article).
\end{rem}

\begin{defn}
    $(u,\alpha) \in \mathscr C(\R, H^1 \oplus L^2)$ is called a mild solution if it satisfies the Duhamel formula
    $$
    (u,\alpha)(t) = \left( e^{i t \Delta}u_0 - i \int_0^t e^{i (t-s) \Delta} A_{\alpha(s)} u(s) \di s , e^{-i t} \alpha_0 - i \int_0^t e^{-i (t-s)} \frac{\mathcal F(\abs{u(s)}^2)}{\abs{k}^{(d-1)/2}} \di s \right).
    $$
\end{defn}
We will also consider solutions $(u,\alpha) \in \mathscr C(\R,L^2 \oplus L^2)$, and still call them mild solutions.
\subsection{Global well-posedness in the energy space $H^1 \oplus L^2$} \label{theorieH1}

\begin{prop} (local existence in $H^1 \oplus L^2$)
For all $R>0$, there exists $T(R)>0$ such that for any $(u_0,\alpha_0) \in H^1 \oplus L^2$ with $\norm{u_0}_{H^1} + \norm{\alpha_0}_{L^2} \leq R$, there exists a unique $(u,\alpha) \in \mathscr C ([0,T(R)],H^1 \oplus L^2)$ mild solution of \eqref{PDE} with $(u(0),\alpha(0)) = (u_0,\alpha_0)$. Furthermore, for all $t \in [0,T(R)],$ the flow map
$$
\phi_t : (u_0,\alpha_0) \in B_{H^1 \oplus L^2}(0,R) \longmapsto  (u,\alpha) \in \mathscr C([0,T(R)], H^1 \oplus L^2)
$$
is Lipschitz.
\end{prop}

\begin{proof}
    Let $T,M >0$ to be fixed later. We will use Banach's fixed point theorem on the complete metric space $(E,d)$, where 
    $$
    E = \{ (u,\alpha) \in  \mathscr C ([0,T],H^1) \times \mathscr C ([0,T],L^2) ; \norm{u}_{L^\infty([0,T],H^1)} + \norm{\alpha}_{L^\infty([0,T],L^2)} \leq M \}, 
    $$
    and 
    $$
    d((u_1,\alpha_1),(u_2,\alpha_2)) = \norm{u_1 - u_2}_{L^\infty([0,T],H^1)} + \norm{\alpha_1 - \alpha_2}_{L^\infty([0,T],L^2)}.
    $$
Consider
$$
\mathscr L (u,\alpha)(t) = \left( e^{i t \Delta}u_0 - i \int_0^t e^{i (t-s) \Delta} A_{\alpha(s)} u(s) \di s , e^{-i t} \alpha_0 - i \int_0^t e^{-i (t-s)} \frac{\mathcal F(\abs{u(s)}^2)}{\abs{k}^{(d-1)/2}} \di s \right).
$$
First, we check that $\mathscr L$ maps $E$ into itself. Using Strichartz estimates for the admissible pair $(4, 2 d /(d-1))$, one finds
    $$\norm{u (t)}_{L^\infty H^1} \leq C \norm{u_0}_{H^1} + C \norm{A_{\alpha} u}_{L^{4/3} W^{1,2 d/(d+1)}}.$$
    Now remark that by the Sobolev embeddings, $A_\alpha \in \dot H^{(d-1)/2}\xhookrightarrow{} L^{2 d}$, $\nabla A_\alpha \in \dot H^{(d-3)/2}\xhookrightarrow{} L^{2 d/3}$ and $u \in H^1 \xhookrightarrow{} L^{2 d/(d-2)}$, hence
    \begin{align*}
    \norm{A_\alpha(t) u(t)}_{W^{1, 2 d/(d+1)}}  & \leq 2 \norm{\mathcal F(\alpha(t) \abs{k}^{(1-d)/2})}_{L^{2 d}} \norm{u(t)}_{H^1}
    \\
    & \leq C \norm {\mathcal F(\alpha(t) \abs{k}^{(1-d)/2})}_{\dot H^{(d-1)/2}}  \norm{u(t)}_{H^1}
    \\
    & = C \norm{\alpha(t)}_{L^2} \norm{u(t)}_{H^1} \\
    &\leq C M^2.
    \end{align*}
By Hölder,
$$
\norm{A_\alpha u}_{L^{4/3} W^{1,2d/(d+1)}} \leq T^{3/4}\norm{A_\alpha(t) u(t)}_{L^\infty W^{1, 2 d/(d+1)}} \leq C T^{3/4} M^2.
$$
Hence the estimate
    \begin{equation} \label{cond1}
    \norm{\mathscr L u }_{L^\infty H^1} \leq C R + C T^{3/4} M^2.
    \end{equation}
    Similarly, one can bound
    \begin{equation} \label{cond2}
            \norm{\mathscr L \alpha}_{L^\infty L^2} \leq R + C T M^2,
    \end{equation}
    where we used the inequality 
    $$\norm{\frac{\mathcal F(\abs{u(s)}^2)}{\abs{k}^{(d-1)/2}}}_{L^2} \leq C \norm{u(s)}_{H^1} \norm{u(s)}_{L^2} \leq C M^2.$$
    Hence choosing first $M = 2(C+1) R$ and then $T$ small enough so that 
    $$ C T M^2 + C T^{3/4} M^2 \leq \frac{M}{2},$$
    we have that $\mathscr L (E) \subset E$. It remains to prove that $\mathscr L$ is, say, $1/2$-Lipschitz. In the same fashion as previously, one can check that 
    \begin{align*}
    d(\mathscr L(u_1,\alpha_1),\mathscr L(u_2,\alpha_2)) &\leq C \norm{A_{\alpha_1} u_1 - A_{\alpha_2} u_2}_{L^{4/3} W^{1,2 d /(d+1)}} + C
    T \norm{\frac{\mathcal F(\abs{u_1}^2-\abs{u_ 2}^2)}{ \abs{k}^{(d-1)/2}}}_{L^2}, \\
    &\leq C T^{3/4} (\norm{\alpha_1}_{L^\infty L^2} + \norm{\alpha_2}_{L^\infty L^2}) \norm{u_1 - u_2}_{L^\infty H^1} \\
    & + C T^{3/4} (\norm{u_1}_{L^\infty H^1} + \norm{u_2}_{L^\infty H^1}) \norm{\alpha_1 - \alpha_2}_{L^\infty L^2} \\ & + C T (\norm{u_1}_{L^\infty H^1} + \norm{u_2}_{L^\infty H^1}) \norm{u_1-u_2}_{L^\infty H^1} \\
    &\leq (2CM T^{3/4} + 2 C T M) d((u_1,\alpha_1),(u_2,\alpha_2)) \\
    &\leq \frac{1}{2}d((u_1,\alpha_1),(u_2,\alpha_2)) ,
    \end{align*}
    for $T$ small enough only depending on $M = 2(C+1) R$. This gives the local existence. For the regularity of the flow map, consider $(u_{0},\alpha_{0}), (v_{0},\beta_{0}) \in B_{H^1 \oplus L^2}(0,R)$ and $(u,\alpha),(v,\beta)$ the associated solutions. Using the former computations, we have for every $t \in [0,T(R)]$,
    \begin{align*}
         \norm{u-v}_{L^\infty([0,t],H^1)}+&\norm{\alpha-\beta}_{L^\infty([0,t],L^2)} \leq C(R) (\norm{u_0-v_0}_{H^1} +\norm{\alpha_0-\beta_0}_{L^2}) \\
         & \ \ \ \ \ \ \ \ \ \ \ \ +\frac{1}{2} (\norm{u-v}_{L^\infty([0,t],H^1)}+\norm{\alpha-\beta}_{L^\infty([0,t],L^2)}),
    \end{align*}
    that is,
    $$
    \norm{u-v}_{L^\infty([0,t],H^1)}+\norm{\alpha-\beta}_{L^\infty([0,t],L^2)} \leq 2 C(R) (\norm{u_0-v_0}_{H^1} +\norm{\alpha_0-\beta_0}_{L^2}),
    $$
    proving the Lipschitz caracter of $\phi_t$.
    \end{proof}

\begin{prop} (conservation laws) \label{conservationH^1}
Let $(u,\alpha) \in \mathscr C([0,T],H^1 \oplus L^2)$ be a solution to \eqref{PDE}. Then, for all $t \in [0,T]$
$$
\begin{aligned}
    \norm{u(t)}_{L^2} = \norm{u(0)}_{L^2}, & \text{ (conservation of the mass)} \\
   h_\infty(u(t),\alpha(t)) = h_\infty(u(0),\alpha(0)), & \text{ (conservation of the energy)}
\end{aligned}
$$
where we recall that the classical energy Hamiltonian $h_\infty:H^1 \oplus L^2 \rightarrow \R$ is defined as
$$
h_\infty(u,\alpha) = \norm{\nabla u}_{L^2}^2 + \norm{\alpha}_{L^2}^2 + \int A_\alpha (x) \abs{u(x)}^2 \di x.
$$
\end{prop}

\begin{proof}
    Let $(u_0,\alpha_0) \in H^1 \oplus L^2$. Let $\chi_m (x) = \chi(x/m)$, with $\chi \in \mathscr C^\infty_c (\R^d,[0,1])$ being identically $1$ in a neighborhood of the origin, and $m \in \N^*$. We define $u_{m,0} := \chi_m(D_x)u_0$ and $\alpha_{m,0} := \chi_m \alpha_0 $. We notice that $u_{m,0} \in H^2$, $\alpha_{m,0} \in \mathcal F H^1$, $\norm{u_{m,0}}_{H^1} \leq \norm{u_0}_{H^1}$, $\norm{\alpha_{m,0}}_{L^2} \leq \norm{\alpha_0}_{L^2}$ and that $(u_{m,0},\alpha_{m,0})$ converges in $B_{H^1 \oplus L^2}(0,\norm{u_0}_{H^1}+\norm{\alpha_0}_{L^2}) $ to $(u_0,\alpha_0)$. \\ \\
    \textbf{Step 1.} Adaptation of the local well-posedness to $H^2 \oplus \mathcal F H^1$. \\
    The proof is almost the same as in the case $H^1 \oplus L^2$. Since $$\norm{\frac{\mathcal F( \abs{u(s)}^2 )}{\abs{k}^{(d-1)/2}}}_{L^2} \lesssim \norm{u(s)}_{H^1} \lesssim \norm{u(s)}_{H^2},$$ we only have to find a bound for
    $$
    \norm{\frac{\mathcal F( \abs{u(s)}^2 )}{\abs{k}^{(d-1)/2}}}_{\mathcal F \dot H^1},
    $$
    as well as
    $$
    \norm{\int_0^t e^{i(t-s)\Delta}  A_{\alpha(s)} u(s) \di s}_{L^\infty H^2} = \norm{\int_0^t e^{i(t-s)\Delta} (1-\Delta) A_{\alpha(s)} u(s) \di s}_{L^\infty L^2},
    $$
    where it suffices to bound $A_\alpha u, A_\alpha (\Delta u), \nabla A_\alpha . \nabla u$ and $ (\Delta A_\alpha) u$ in some $L^{p'} L^{q'}$ norm, where $(p,q)$ is a chosen admissible pair. For the first one, we use Hardy-Littlewood-Sobolev and the embedding $H^2 \hookrightarrow L^{4 d/(2d-3)}$ to estimate
    $$
    \norm{\frac{\mathcal F( \abs{u(s)}^2 )}{\abs{k}^{(d-1)/2}}}_{\mathcal F \dot H^1} = \norm{\frac{\mathcal F( \abs{u(s)}^2 )}{\abs{k}^{(d-3)/2}}}_{L^2} = \norm{ \frac{c_d}{\abs{x}^{(d+3)/2}} \ast \abs{u(s)}^2}_{L^2} \lesssim \norm{u(s)}_{L^{4d/(2d-3)}}^2 \lesssim \norm{u(s)}_{H^2}^2.
    $$
    For the second bound, we claim that $A_\alpha u, A_\alpha (\Delta u), \nabla A_\alpha . \nabla u$ and $ (\Delta A_\alpha) u$ all belong to $L^{2d/(d+1)}$, and thus taking the admissible pair $(p,q) = (4, 2d/(2d-1))$ as before gives the appropriate bound. Indeed, keeping in mind that $d \geq 3$ and $\alpha \in \mathcal F H^1$, one has
    $$\begin{cases}
        A_\alpha \in \dot H^{(d-1)/2} \cap \dot H^{(d+1)/2} \hookrightarrow L^{2 d}, \\
        \nabla A_\alpha \in  \dot H^{(d-3)/2} \cap \dot H^{(d-1)/2} \hookrightarrow  L^{2 d}, \\
        \Delta A_\alpha \in \dot H^{(d-5)/2} \cap \dot H^{(d-3)/2} \hookrightarrow L^{2d/3}.
    \end{cases}$$
    Therefore, using Sobolev embeddings,
    $$
    \begin{cases}
        u, \nabla u, \Delta u \in L^2 \Rightarrow  A_\alpha u,\nabla u.\nabla A_\alpha, A_\alpha \Delta u \in L^{2d/(d+1)}, \\
        u \in L^{2d/(d-2)} \Rightarrow \Delta A_\alpha u \in L^{2d/(d+1)},
    \end{cases}
    $$
    which concludes. Therefore, $(u,\alpha)$ is also a strong solution of the Landau-Pekar equations, meaning that $\eqref{PDE}$ holds in $L^2 \oplus L^2$.
\\
    
\textbf{Step 2.} Proof of the conservation laws. \\
Let $(u_m,\alpha_m)$ be the local solution in $H^2 \oplus \mathcal F H^1$ of \eqref{PDE} with initial condition $(u_{m,0},\alpha_{m,0})$. By regularity of $(u_m,\alpha_m)$, we can differentiate
\begin{align*}
    \frac{\di}{\di t}\norm{u_m(t)}_{L^2}^2 &= 2 \mathrm{Re} \langle \partial_t u_m(t),  u_m(t) \rangle \\
    &= 2 \mathrm{Re} \langle  i \Delta u_m(t) - i A_{\alpha_m(t)}u_m(t) , u_m(t) \rangle \\
    &=  2 \mathrm{Re} (i \norm{\nabla u_m(t)}_{L^2}^2) + 2 \mathrm{Re} \left(i \int A_{\alpha_m(t)} \abs{u_m(t)}^2 \right) \\
    &=0,
\end{align*}
so that for every $t \in (0,T)$,$ \norm{u_m(t)}_{L^2}$ is constant, and by continuity of $u_m$, also for every $t \in [0,T]$, giving
$$
\norm{u_m(t)}_{L^2} = \norm{u_{m,0}}_{L^2}.
$$
By continuity of the flow in $B_{H^1 \oplus L^2}(0,\norm{u_0}_{H^1}+\norm{\alpha_0}_{L^2})$, we can pass to the limit and obtain the conservation of the mass.
Similarly, we can differentiate the Hamiltonian and use the fact that $(u_m,\alpha_m)$ solves \eqref{PDE} to write
\begin{align*}
    \frac{\di}{\di t} h_\infty(u_m(t),\alpha_m(t))&= 2\mathrm{Re} \langle \partial_t u_m(t), -\Delta u_m(t) \rangle + 2 \mathrm{Re} \langle \partial_t \alpha_m(t), \alpha_m(t) \rangle \\
    &+ 2 \mathrm{Re} \int A_{\alpha_m(t)} \partial_t u_m(t) \overline{u_m(t)} + \int A_{\partial_t \alpha_m(t)} \abs{u_m(t)}^2 \\
    &=  2\mathrm{Re} \langle- i A_{\alpha_m(t)} u_m(t), -\Delta u_m(t) \rangle + 2 \mathrm{Re} \langle -i \frac{\mathcal F(\abs{u_m(t)}^2)}{ \abs{k}^{(d-1)/2}}, \alpha_m(t) \rangle \\
    &+ 2 \mathrm{Re} \int A_{\alpha_m(t)} (i \Delta u_m(t)) \overline{u_m(t)} + 2 \mathrm{Re} \int \frac{- i \alpha_m(t)}{\abs{k}^{(d-1)/2}} \mathcal F(\abs{u_m(t)}^2) \\
    &=0.
\end{align*}
where from the first to second line, we have erased the terms of the form $\mathrm{Re} (i x)$ with $x \in \R$ and used Parseval's inequality to transfer the Fourier transform to the terms in $\abs{u_m(t)}^2$. We are able to conclude in the same way as for the conservation of mass, if we prove that $h_\infty : H^1 \oplus L^2 \rightarrow \R$ is continuous.
\\ \\
\textbf{Step 3.} $h_\infty : H^1 \oplus L^2 \rightarrow \R$ is continuous. \\
Since $(u,\alpha) \in H^1 \oplus L^2 \longmapsto \norm{\nabla u}_{L^2}^2+\norm{\alpha}_{L^2}^2 \in \R$ is continuous, we only need to prove that 
$$
(u,\alpha) \in H^1 \oplus L^2 \longmapsto \int A_\alpha \abs{u}^2
$$
is. Recall the Gagliardo-Nirenberg inequality:
$$
\norm{u}_{L^{4d/(2d-1)}} \leq C \norm{\nabla u}_{L^2}^{1/4} \norm{u}_{L^2}^{3/4}.
$$
For any $(u,\alpha),(v,\beta) \in H^1 \oplus L^2$,
\begin{align*}
    \abs{ \int A_\alpha \abs{u}^2-\int A_\beta \abs{v}^2}  &\leq \norm{A_\alpha \abs{u}^2 - A_\beta \abs{v}^2}_{L^1} \\
        & \leq \norm{A_\alpha-A_\beta}_{L^{2 d}} \norm{u^2}_{L^{2 d/(2d-1)}} +\norm{A_\beta}_{L^{2 d}} \norm{u \overline{u} - v \overline{v}}_{L^{2 d/(2d-1)}}  \\
        & \leq \norm{A_\alpha-A_\beta}_{L^{2 d}} \norm{u}_{L^{4 d/(2d-1)}}^2 \\& + \norm{A_\beta}_{L^{2 d}} \norm{u  - v}_{L^{4 d/(2d-1)}} (\norm{u}_{L^{4 d/(2d-1)}} +\norm{v}_{L^{4 d/(2d-1)}} ) \\
        & \leq C \norm{\alpha-\beta}_{L^2} \norm{\nabla u}_{L^2}^{1/2} \norm{u_0}_{L^2}^{3/2}
        \\ &+ C \norm{\beta}_{L^2} \norm{\nabla(u-v)}_{L^2}^{1/4} \norm{u-v}_{L^2}^{3/4} (\norm{\nabla u}_{L^2}^{1/4} \norm{u}_{L^2}^{3/4} +\norm{\nabla v}_{L^2}^{1/4} \norm{v}_{L^2}^{3/4}),
\end{align*}
proving the continuity.
\end{proof}

\begin{prop}(global well posedness in $H^1 \oplus L^2$)
    Let $(u_0,\alpha_0) \in H^1 \oplus L^2$. Then the solution $(u,\alpha)$ of \eqref{PDE} starting from $(u_0,\alpha_0)$ is global in time.
\end{prop}

\begin{proof}
   Using the continuity of $h$ and the conservation of mass, we obtain the following bound:
    \begin{align*}
        \abs{ \int A_\alpha(t)  \abs{u(t)}}  
        & \leq C \norm{\alpha(t)}_{L^2} \norm{\nabla u(t)}_{L^2}^{1/2} \norm{u_0}_{L^2}^{3/2} \\
        &\leq C \norm{\alpha(t)}_{L^2} \norm{\nabla{u}(t)}^{1/2}_{L^2}\\
        &\leq (\eps \norm{\nabla{u}(t)}_{L^2} + C_\eps)\norm{\alpha(t)}_{L^2} \\
        &\leq \eps (\norm{\nabla{u}(t)}_{L^2}^2 + \norm{\alpha(t)}_{L^2}^2) + \tilde C_\eps.
    \end{align*}
    Choosing $\eps = \frac{1}{2}$, one obtains
    $$ \norm{u(t)}_{H^1}^2 + \norm{\alpha(t)}_{L^2}^2 \leq 2 (h_\infty(u(t),\alpha(t)) + \norm{u(t)}_{L^2}^2+\tilde C_{1/2}) = 2( h(u_0,\alpha_0) + \norm{u_0}_{L^2}^2 +\tilde C_{1/2})=:2 R^2. $$
        Now, consider the maximal time of existence $T$ of $(u,\alpha)$. Assume that $T<\infty$. Then we can construct a solution of \eqref{PDE} starting at time $T-\frac{1}{2} T(R)$ with initial condition $(u(T-\frac{1}{2} T(R)),\alpha(T-\frac{1}{2} T(R)))$. Since the norm of the initial condition is bounded by $R$, and by uniqueness, this solution extend $u$ to $[0,T+\frac{1}{2}T(R))$, which contradicts the maximality of $T$. Hence $T=\infty$. 
\end{proof}

\subsection{Global well-posedness in the physical space $L^2 \oplus L^2$} \label{theorieL2}

\begin{prop} (local existence in $L^2 \oplus L^2$)
For all $R>0$, there exists $T(R)>0$ such that for any $(u_0,\alpha_0) \in L^2 \oplus L^2$ with $\norm{u_0}_{L^2} + \norm{\alpha_0}_{L^2} \leq R$, there exists a unique $(u,\alpha) \in \mathscr C ([0,T(R)],L^2) \cap L^2([0,T(R)],L^{2d/(d-2)}) \times \mathscr C([0,T(R)],L^2)$ solution of \eqref{PDE} with $(u(0),\alpha(0)) = (u_0,\alpha_0)$. Furthermore, for all $t \in [0,T(R)],$ the flow map
$$
\phi_t : \begin{array}{ll}  B_{L^2 \oplus L^2}(0,R) &\longrightarrow  \mathscr C ([0,T(R)],L^2) \cap L^2([0,T(R)],L^{2d/(d-2)}) \times \mathscr C([0,T(R)],L^2) \\ (u_0,\alpha_0) &\longmapsto (u,\alpha).
\end{array}
$$
is Lipschitz.
\end{prop}

\begin{proof}
We use the same techniques as in the case $H^1 \oplus L^2$. Let $T,M >0$ to be fixed later. Consider the complete metric space $(E,d)$, where 
    \begin{align*}
    E = \{ (u,\alpha) \in  (\mathscr C ([0,T],L^2) & \cap L^2([0,T],L^{2d/(d-2)})  )\times \mathscr C ([0,T],L^2) ; \\ &\norm{u}_{L^\infty([0,T],L^2)}+ \norm{u}_{L^2([0,T],L^{2d/(d-2)})} + \norm{\alpha}_{L^\infty([0,T],L^2)} \leq M \}, 
    \end{align*}
    and 
    $$
    d((u_1,\alpha_1),(u_2,\alpha_2)) = \norm{u_1 - u_2}_{L^\infty([0,T],L^2)} +  \norm{u_1 - u_2}_{L^2([0,T],L^{2d/(d-2)})}+ \norm{\alpha_1 - \alpha_2}_{L^\infty([0,T],L^2)}.
    $$
    We define $\mathscr L$ identically as in the case $H^1 \oplus L^2$. First, $\mathscr L (E) \subset E$. indeed, using the same proof as before (but with no space derivative), one gets
    $$
    \norm{u}_{L^\infty([0,T],L^2)}+ \norm{u}_{L^2([0,T],L^{2d/(d-2)})} \leq C R + C T^{3/4} M^2.
    $$
    For the other field, it is more delicate. We first express
    $$
    \frac{1}{\abs{k}^{(d-1)/2}} \mathcal F(\abs{u}^2) = \mathcal F \left( \frac{c_d}{\abs{x}^{(d+1)/2}} \ast \abs{u}^2 \right).
    $$
    Since the $L^2$ norm preserves the Fourier transform, it suffices to bound 
    $$\norm{ \frac{1}{\abs{x}^{(d+1)/2}} \ast \abs{u}^2}_{L^2}.$$
    By Hardy-Littlewood-Sobolev, the latter is estimated by $\norm{u^2}_{L^r},$ where 
    $$
    1+\frac{1}{2} = \frac{d+1}{2 d} + \frac{1}{r},
    $$
    that is, $r = \frac{2 d}{2 d -1}$. Recalling the expression for $\mathscr L \alpha (t)$:
    $$
    \mathscr L \alpha (t) = e^{-i t} \alpha_0 - i \int_0^t e^{-i (t-s)} \frac{\mathcal F(\abs{u(s)}^2)}{\sqrt{2 \abs{k}^{d-1}}} \di s,
    $$
    we end up with
    $$
    \norm{\alpha(t)}_{L^2} \leq \norm{\alpha_0}_{L^2} + C \int_0^t \norm{u(s)}^2_{L^{4 d /(2d -1)}} \di s.
    $$
    It suffices now to bound $\norm{u}^2_{L^2([0,t],L^{4d/(2d-1)})} \leq \norm{u}^2_{L^2([0,T],L^{4d/(2d-1)})}.$ A straigthforward computation shows that the pairs $(8,4d/(2d-1))$ and $(2, 2d/(d-2))$ are admissible. Since $2 < 4d/(2d-1)<2d/(d-2)$, we can use the interpolation inequality to write
    $$
    \norm{u(t)}_{L^{4d/(2d-1)}} \lesssim \norm{u(t)}_{L^2}^\theta \norm{u(t)}^{1-\theta}_{L^{2d /(d-2)}} \lesssim \norm{u}_{L^\infty L^2}^\theta \norm{u(t)}_{L^{2d/(d-2)}}^{1-\theta},
    $$
    where 
    $$
    \frac{2d-1}{4d} = \frac{\theta}{2} + (1-\theta) \frac{d-2}{2d},
    $$
    that is, $\theta = 3/4$.Thus,
    \begin{align*}
     \norm{u}^2_{L^2([0,T],L^{4d/(2d-1)})} &= \int_0^T \norm{u(s)}^2_{L^{4d/(2d-1)}} \di s \\ &\lesssim \norm{u}^{3/2}_{L^\infty L^2} \int_0^T \norm{u(s)}^{1/2}_{L^{2d/(d-2)}} . 1\di s \\
     &\lesssim \norm{u}_{L^\infty L^2}^{3/2} \norm{  \norm{u(s)}^{1/2}_{L^{2d/(d-2)}}}_{L^4([0,T])} \norm{1}_{L^{4/3}([0,T])} \\
     & \lesssim  \norm{u}_{L^\infty L^2}^{3/2} \norm{u}^{1/2}_{L^2 L^{2d/(d-2)}} T^{3/4}.
    \end{align*}
    Hence, we have
    $$
    \norm{\alpha}_{L^\infty L^2} \leq \norm{\alpha_0}_{L^2}+ T^{3/4} \norm{u}^{3/2}_{L^\infty L^2} \norm{u}^{1/2}_{L^2 L^{2 d/(d-2)}} \leq R+ C M^2 T^{3/4}.
    $$
    Choosing $M(R)$ large enough and then $T(R)$ small enough, we conclude that $\mathscr L$ maps $E$ to $E$. Similar techniques allow us to prove that $\mathscr L$ is $1/2$-Lipschitz for $T(R)$ perhaps smaller, proving the local existence. The proof for the Lipschitz character of the flow map is led in the same way as in the case $H^1 \oplus L^2$.
    \end{proof}

\begin{prop}(conservation of the mass) \label{conservationmassL2}
Let $(u,\alpha) \in \mathscr C ([0,T],L^2) \cap L^2([0,T],L^{2d/(d-2)}) \times \mathscr C([0,T], L^2)$ be a solution of \eqref{PDE}. Then for all $t \in [0,T]$, 
$$
\norm{u(t)}_{L^2} = \norm{u_0}_{L^2}.
$$
\end{prop}

\begin{proof}
    It is the same as in the proof of Proposition \ref{conservationH^1}, noticing that the approximation $u_m$ converges to $u_{m,0}$ also in $L^2$, and using the continuity of the flow map in $B_{L^2 \oplus L^2}(0,R)$, where $R = \norm{u_0}_{L^2}+\norm{\alpha_0}_{L^2}$.
\end{proof}

\begin{prop} (global existence in $L^2 \oplus L^2$)
For every $(u_0,\alpha_0) \in L^2 \oplus L^2$,there exists $(u,\alpha) \in \mathscr C (\R,L^2) \cap L^2(\R,L^{2d/(d-2)}) \times \mathscr C(\R, L^2)$ solution of \eqref{PDE} with $(u(0),\alpha(0)) = (u_0,\alpha_0)$.
\end{prop}

\begin{proof}
    Let $(u_0,\alpha_0) \in L^2 \oplus L^2$ and $(u,\alpha)$ the maximal solution starting from $(u_0,\alpha_0)$. We assume for a contradiction that the maximal time of existence $T_{\mathrm{max}} <\infty$. We first prove that the norm of $(u,\alpha)$ is forbidden to explode. We go back to the estimates in the proof of the local existence: for $t \in [0,T_{\mathrm{max}})$,
    \begin{align*}
       \norm{u}_{L^\infty([0,t], L^2)} + \norm{u}_{L^2([0,t],L^{2 d/(d-2)})} &\leq C \norm{u_0}_{L^2} + C \norm{\alpha}_{L^\infty([0,t], L^2)} \norm{u}_{L^\infty([0,t], L^2)} T_{\mathrm{max}}^{3/4}, \\
       \norm{\alpha}_{L^\infty([0,t],L^2)} &\leq \norm{\alpha_0}_{L^2} + C \norm{u}_{L^\infty([0,t], L^2)}^{3/2} \norm{u}_{L^2([0,t],L^{2 d/(d-2)})}^{1/2} T_{\mathrm{max}}^{3/4}.
    \end{align*}
    By the conservation of mass given by Proposition \ref{conservationmassL2},
    $$
    \norm{u}_{L^\infty([0,t],L^2)} = \norm{u_0}_{L^2}.
    $$
    Now, plugging back $\norm{\alpha}_{L^\infty([0,t],L^2)}$ in the estimate of $\norm{u}_{L^\infty([0,t], L^2)} + \norm{u}_{L^2([0,t],L^{2 d/(d-2)})}$, we get for all $t \in [0,T_{\mathrm{max}})$ that
    $$
    \norm{u}_{L^2([0,t],L^{2 d/(d-2)})} \leq C \norm{u_0}_{L^2}+ C \norm{u_0}_{L^2} T_{\mathrm{max}}^{3/4}  (\norm{\alpha_0}_{L^2}+ C \norm{u_0}_{L^2}^{3/2} \norm{u}_{L^2([0,t],L^{2 d/(d-2)})}^{1/2} T_{\mathrm{max}}^{3/4}),
    $$
    or in other words,
    $$
     \norm{u}_{L^2([0,t],L^{2 d/(d-2)})} \lesssim_{\norm{u_0}_{L^2},\norm{\alpha_0}_{L^2}, T_{\mathrm{max}}} (1+ \norm{u}_{L^2([0,t],L^{2 d/(d-2)})}^{1/2}).
    $$
    Therefore,
    $$
     \underset{t \in [0,T_{\mathrm{max}})} \sup \norm{u}_{L^2([0,t],L^{2 d/(d-2)})} < \infty,
    $$
    which implies
    $$
     \underset{t \in [0,T_{\mathrm{max}})} \sup \norm{\alpha}_{L^\infty([0,t],L^2)} <\infty.
    $$
    Let $R := \norm{u_0}_{L^2} +  \underset{t \in [0,T_{\mathrm{max}})} \sup \norm{\alpha}_{L^\infty([0,t],L^2)}$. By the local existence in $L^2 \oplus L^2$, we can consider the solution starting from $(u(T_{\mathrm{max}}-T(R)/2),\alpha(T_{\mathrm{max}}-T(R)/2))$ at time $T_{\mathrm{max}}-T(R)/2$. Constructed this way, the latter extend by uniqueness the solution $(u,\alpha)$ on a time of existence $[0,T_{\mathrm{max}}+T(R)/2)$, yielding a contradiction by maximality of $T_{\mathrm{max}}$. Therefore $T_{\mathrm{max}}= + \infty$ and the solution is global.  
\end{proof}

Subsection \ref{theorieH1} and \ref{theorieL2} respectively prove Proposition \ref{thmgwp} and Theorem \ref{thmgwpl2}, as presented in the Introduction.

\subsection{Classical dressing transform}

Nothing is new in this subsection from \cite[Subsection 3.1]{ammari2017sima}, but we recall the results for clarity. The expression of the classical dressing generator is
$$
\mathscr D_{i B_\infty}(u,\alpha) = \iint (i B_\infty(k) e^{-i k . x} \overline{\alpha}(k) - i B_\infty(k) e^{i k . x} \alpha(k) )  \abs{u(x)}^2 \di x \di k.
$$
It gives rise to a set of Hamilton equations 
\begin{equation} \label{PDEdressing}
\begin{cases}
    i \partial_\theta u = A_{\alpha, i B_\infty} u, \\
    i \partial_\theta \alpha =i B_\infty \mathcal F (\abs{u}^2 ),  
\end{cases}
\end{equation}
where
$$
A_{\alpha, i B_\infty} (x) = \int (i B_\infty(k) e^{-i k . x} \overline{\alpha}(k) - i B_\infty(k) e^{i k .x} \alpha(k) ) \di k. 
$$
This system happens to be exactly solvable and the solutions $(u_\theta,\alpha_\theta)$ have the following form:
\begin{equation} \label{PDEdressingsolution}
\begin{cases}
    u_\theta (x) = u_0(x) \exp \left(  - i\theta A_{\alpha , i B_\infty} \right), \\
    \alpha_\theta (k) = \alpha_0(k) + \theta B_\infty(k) \mathcal F(\abs{u}^2)(k).    
\end{cases}
\end{equation}

\begin{rem}
    The scalar field $A_\alpha$ defined for writing the Landau-Pekar equations can be rewritten consistently with these notations as
    $$
    A_\alpha = A_{\alpha, f_\infty}.
    $$
\end{rem}

\begin{prop} \cite[Proposition 3.2]{ammari2017sima} \label{gwpdressing}
    For every $(u_0,\alpha_0) \in L^2 \oplus L^2$, there exists a unique $(u,\alpha) \in \mathscr C(\R,L^2 \oplus L^2)$ solution of \eqref{PDEdressing} with $(u(0),\alpha(0)) = (u_0,\alpha_0)$. Moreover, the dressing flow, defined for every $\theta \in \R$, as the map
    $$
    \mathrm{D}_{i B_\infty}(\theta) : (u_0,\alpha_0) \in L^2 \oplus L^2 \mapsto (u_\theta, \alpha_\theta) \in L^2 \oplus L^2,
    $$
    is a symplectomorphism, and one has the propagation of regularity
    $$
    \mathrm{D}_{i B_\infty}(\theta) (H^1 \oplus L^2) \subset H^1 \oplus L^2.
    $$
\end{prop}

\subsection{Linking the dressed and undressed systems}

Let us first recall the definitions of the undressed and dressed classical Hamiltonian functionals, respectively $h_\infty$ and $\hat h_\infty$. For $(u,\alpha) \in H^1 \oplus L^2$,
\begin{align*}
    h_\infty(u,\alpha) &= \norm{\nabla u}_{L^2}^2 + \norm{\alpha}_{L^2}^2 + \int A_\alpha (x) \abs{u(x)}^2 \di x,\\
    \hat h_\infty (u,\alpha) &= \norm{\nabla u}_{L^2}^2 + \norm{\alpha}_{L^2}^2 +  \int 2 \mathrm{Re} \langle \alpha, f_{\sigma_0} e^{-ik.x} \rangle \abs{u(x)}^2  \di x \\
        &+ \iint \abs{u(x)}^2 V_\infty(x-y) \abs{u(y)}^2 \di x \di y \\
        &+ \int (2 \mathrm{Re} \langle \alpha, k B_\infty e^{-ik.x} \rangle)^2 \abs{u(x)}^2 \di x  \\
        & - 2 \int   \overline{u(x)} ( \langle \alpha, k B_\infty e^{-i k . x} \rangle  .D_{x} + D_x . \langle k B_\infty e^{-i k .x} , \alpha \rangle) u(x) \di x.
    \end{align*}
Thanks to \textbf{Step 3.} in Proposition \ref{conservationH^1}, $h_\infty$ is well defined on $H^1 \oplus L^2$. Since $f_{\sigma_0}, k B_\infty \in L^2$ and $V_\infty \in L^\infty$, it is clear that $\hat h_\infty$ is also well defined on $H^1 \oplus L^2$. The next result precise their link through the classical dressing flow $\mathrm{D}_{i B_\infty}(1)$ at time $1$.

\begin{prop}
On $H^1 \oplus L^2$, $\hat h_\infty = h_\infty \circ \mathrm{D}_{i B_\infty}(1).$
\end{prop}

\begin{proof}
    We start from the right-hand-side. Let $(u,\alpha) \in H^1 \oplus L^2$. Then,
    \begin{align*}
        h_\infty(\mathrm{D}_{i B_\infty }(1)(u,\alpha)) &= h_\infty ( u e^{- i A_{\alpha , i B_\infty}}, \alpha +  B_\infty\mathcal F(\abs{u}^2))
        \\
        &= \norm{\nabla u  e^{- i A_{\alpha , i B_\infty}}  - i u A_{\alpha, k B_\infty} e^{- i A_{\alpha , i B_\infty}} }_{L^2}^2 + \norm{\alpha + B_\infty \mathcal F(\abs{u}^2)}_{L^2}^2 \\
        &+ \int A_{\alpha + B_\infty \mathcal F(\abs{u}^2)} \abs{u e^{- i A_{\alpha , i B_\infty}}}^2 \\
        &= \norm{\nabla u}^2_{L^2} + \norm{u A_{\alpha,k B_\infty}}^2_{L^2} + 2 \mathrm{Re} \langle \nabla u, - i u A_{\alpha, k B_\infty} \rangle \\
        &+\norm{\alpha}_{L^2}^2 + \norm{B_\infty \mathcal F(\abs{u}^2)}_{L^2}^2 + 2 \mathrm{Re} \langle \alpha, B _\infty \mathcal F(\abs{u}^2) \rangle \\
        &+ \int 2 \mathrm{Re} \langle \alpha, f_\infty e^{-i k .x} \rangle \abs{u}^2 \di x + \int2 \mathrm{Re} \langle B_\infty  \mathcal F(\abs{u}^2), f_\infty e^{-i k .x} \rangle \abs{u}^2 \di x.
    \end{align*}
The term of degree $4$ in $u$ give us the $V_\infty$ term:
\begin{align*}
     \norm{B_\infty \mathcal F(\abs{u}^2)}_{L^2}^2 &+ \int2 \mathrm{Re} \langle B_\infty  \mathcal F(\abs{u}^2), f_\infty e^{-i k .x} \rangle \abs{u}^2 \di x \\ &= \mathrm{Re} \langle (\abs{B_\infty}^2 +2B_\infty f_\infty) \mathcal F(\abs{u}^2), \mathcal F(\abs{u}^2) \rangle \\
     &=\mathrm{Re} \iiint \abs{u(x)}^2 \abs{u(y)}^2  (\abs{B_\infty (k)}^2 +2B_\infty (k) f_\infty (k)) e^{-i k .(x-y)} \di x \di y \\
     &= \iint \abs{u(x)}^2 V_\infty(x-y) \abs{u(y)}^2 \di x \di y.
\end{align*}
Now,
\begin{align*}
    \norm{u A_{\alpha, k B_\infty}}^2_{L^2} = \int \abs{u(x)}^2 (2 \mathrm{Re} \langle \alpha, k B_\infty e^{-i k .x} \rangle)^2 \di x.
\end{align*}
Since $B_\infty + \abs{k}^2 B_\infty + f_\infty = f_{\sigma_0}$, it suffices now to show that 
\begin{align*}
2 \mathrm{Re} \langle \nabla u, - i u A_{\alpha, k B_\infty} \rangle &= \int 2 \mathrm{Re} \langle \alpha, \abs{k}^2 B_\infty e^{-i k .x} \rangle \abs{u}^2 \di x \\ &- 2  \langle u, \left(\langle \alpha, kB_\infty e^{- i k .x} \rangle .D_x + D_x.\langle k B_\infty e^{-i k . x} , \alpha \rangle \right) u \rangle.
\end{align*}
To this end, we view $\langle \alpha, kB_\infty e^{- i k .x} \rangle$ and its conjugate as a multiplication operator and a simple commutation computation show that, as operators,
\begin{align*}
    D_x (\langle \alpha, kB_\infty e^{- i k .x} \rangle+ \langle  kB_\infty e^{- i k .x}, &\alpha \rangle)+ (\langle \alpha, kB_\infty e^{- i k .x} \rangle+ \langle  kB_\infty e^{- i k .x}, \alpha \rangle) D_x \\
    &=2 (\langle \alpha, k B_\infty e^{-i k . x} \rangle D_x + D_x \langle k B_\infty e^{-i k.x} , \alpha \rangle) \\
    &- \langle \alpha, \abs{k}^2 B_\infty e^{- i k .x} \rangle -\langle  \abs{k}^2 B_\infty e^{- i k .x}  ,\alpha \rangle \\
    &=2 (\langle \alpha, k B_\infty e^{-i k . x} \rangle D_x + D_x \langle k B_\infty e^{-i k.x} , \alpha \rangle) - A_{\alpha , \abs{k}^2 B_\infty}(x).
\end{align*}
Therefore,
\begin{align*}
    2 \mathrm{Re} \langle \nabla u, - i u A_{\alpha, k B_\infty} \rangle  &= - 2 \mathrm{Re} \langle u, D_x ( A_{\alpha, kB_\infty} u) \rangle \\
    &= - \langle u, (2 (\langle \alpha, k B_\infty e^{-i k . x} \rangle D_x + D_x \langle k B_\infty e^{-i k.x} , \alpha \rangle) - A_{\alpha , \abs{k}^2 B_\infty}(x))   u \rangle,
\end{align*}
which concludes since
$$
\langle u, A_{\alpha, \abs{k}^2 B_\infty} u \rangle = \int \abs{u(x)}^2 2 \mathrm{Re} \langle \alpha, \abs{k}^2 B_\infty e^{-i k . x} \rangle \di x.
$$
\end{proof}

\begin{cor} The Hamilton equations generated by $\hat h_\infty$ are globally well posed in $\mathscr C (\R, H^1\oplus L^2)$. More specifically, for every $z_0 = (u_0,\alpha_0) \in H^1 \oplus L^2$, there exists a unique $z = (u,\alpha) \in \mathscr C(\R, H^1 \oplus L^2)$ with $z(0) = (u(0),\alpha(0)) = (u_0,\alpha_0) = z_0$ satisfying the Duhamel formula
$$
z(t) = \phi^0_t( z_0)  - i \int_0^t \phi^0_{t-s}  \circ \nabla_{\overline{z}} \hat h_\infty (z(s)) \di s.
$$
Furthermore, the corresponding flow $\hat \phi_t$ satisfies, for every $t \in \R$,
    $$
    \mathrm{D}_{i B_\infty}(-1) \circ \hat \phi_t \circ  \mathrm{D}_{i B_\infty}(1) = \phi_t.
    $$
\end{cor}

\begin{rem}
    Here $\phi^0_t$ is the Hamiltonian flow generated by the free classical Hamiltonian $h_0(u,\alpha) = \norm{\nabla u}_{L^2}^2 + \norm{\alpha}_{L^2}^2$. Its action is very explicit and takes the form
$$
\phi_t^0 (u,\alpha) = (e^{i t \Delta} u , e^{- i t } \alpha ).
$$
\end{rem}

\begin{proof} \label{linkflows}
    A direct computation shows that $z(t)$ solves the Duhamel formula above if and only if $w(t) = \mathrm{D}_{i B_\infty}(1) z(t)$ solves the Duhamel formula for the undressed evolution, that is, 
    $$
w(t) = \phi^0_t( w_0)  - i \int_0^t \phi^0_{t-s}  \circ \nabla_{\overline{z}}  h_\infty (w(s)) \di s.
$$
This yields, by the regularity properties of the dressing and the global well-posedness of the undressed system, the global well-posedness of the dressed evolution. The relation $\mathrm{D}_{i B_\infty}(-1) \circ \hat \phi_t \circ  \mathrm{D}_{i B_\infty}(1) = \phi_t$ is deduced from $\hat h_\infty = h_\infty \circ \mathrm{D}_{i B_\infty}$.
\end{proof}

\subsection{The interaction picture for the dressed PDE}

In this subsection we look at the interaction picture, that is the Hamiltonian flow generated by the vector field
$$
X(t,\cdot) = - i \phi^0_{-t} \circ \nabla_{\overline{z}} \hat h_{\infty,I} \circ \phi^0_t,
$$
as it will be the one arising from the classical limit of the quantum mechanical model. 
\begin{prop}
    The following statements are equivalent:
\begin{enumerate}
    \item $z = (u,\alpha) \in \mathscr C(\R,H^1 \oplus L^2)$ satisfies the Duhamel formula: for all $t \in \R$,
    $$
    z(t) = \phi^0_t( z_0)  - i \int_0^t \phi^0_{t-s}  \circ \nabla_{\overline{z}} \hat h_\infty (z(s)) \di s.
    $$
    \item $w: t \mapsto \phi_{-t}^0 ( z(t) )$ belongs to $\mathscr C(\R,H^1 \oplus L^2)$ and satisfies the Duhamel formula: for all $t \in \R$,
    $$
    w(t) = w_0 + \int_0^t X(s,w(s)) \di s.
    $$
\end{enumerate}
\end{prop}

\begin{proof}
    It is an explicit computation using the linearity of the free flow $\phi^0_t$.
\end{proof}

\section{Derivation of the classical system} \label{Derivation of the classical system}
This section is devoted to the proof of the mean-field limit of the evolution given by the quantum Hamiltonian $H_\infty$ towards the classical evolution governed by the classical Hamiltonian functional $\hat h_\infty$. Due to the easier description of the domain of the dressed Hamiltonian $\hat H_\infty$, most of our analysis will in fact concern the dressed evolutions, both quantum and classical.

\subsection{Wigner measures}
We first recall the definitions of normal states and Wigner measures for clarity.
\begin{defn}
    A normal state, or density matrix, is a self-adjoint operator $\rho$ acting on a Hilbert space $\mathscr F$ such that $\rho \geq 0$ and $\Tr(\rho) = 1.$ The set of density matrices is denoted $\mathfrak S^1_{+,1}(\mathscr F)$.
\end{defn}

\begin{defn}
    Let $\mathscr X$ be a separable Hilbert space representing the classical phase-space. Consider a family of density matrices $(\rho_\eps)_{\eps \in (0,1]} \subset \mathfrak S^1_{+,1} (\mathscr F(\mathscr Z))$. We say that $\mu \in \mathscr P(\mathscr Z)$ is a Wigner, or semiclassical measure for $(\rho_\eps)_{\eps \in (0,1]}$ if there exists $\eps_n \rightarrow 0$ such that for every $f \in \mathscr Z$,
    $$
    \Tr( W_{\eps_n} (f) \rho_{\eps_n}) \underset{n \longrightarrow \infty} \longrightarrow \int_{\mathscr Z}  e^{\sqrt{2} i \mathrm{Re} \langle z, f\rangle}    \di \mu( z).
    $$
    We denote by $\mathscr M ((\rho_\eps)_{\eps \in (0,1]} )$ the set of Wigner measures for the family $(\rho_\eps)_{\eps \in (0,1]}$, and will sometimes write $\rho_{\eps_n} \longrightarrow \mu$.
\end{defn}
The following is a sufficient condition for the existence of a Wigner measure for a given family of density matrices. 
\begin{prop} \cite[Theorem 6.2]{ammari2008ahp}
    Consider a family of density matrices $(\rho_\eps)_{\eps \in (0,1]} \subset \mathfrak S^1_{+,1} (\mathscr F(\mathscr Z))$ such that there exists $\delta>0$ with
    $$
    \underset{\eps \in (0,1]} \sup \Tr( (\di \Gamma_\eps(\Id))^\delta \rho_\eps) <\infty.
    $$
    Then 
    $$
    \mathscr M ((\rho_\eps)_{\eps \in (0,1]} ) \neq \emptyset.
    $$
    Furthermore, if $\mu \in \mathscr M ((\rho_\eps)_{\eps \in (0,1]})$ and $\eps_n \rightarrow 0$ is the associated sequence,
    $$
    \int_{\mathscr Z} \norm{z}_\mathscr Z^{ 2 \delta} \mu(\di z) \leq \underset{ n \rightarrow \infty} \liminf \ \Tr( (\di \Gamma_{\eps_n}(\Id))^\delta \rho_{\eps_n})<\infty.
    $$
\end{prop}

We now state a preparatory lemma giving useful estimates for the relative bound of the Weyl operators with respect to second quantized operators.

\begin{lem} \label{WxiQH0}
    If $\xi \in H^1 \oplus L^2$, then $W(\xi) Q(H_0) \subset Q(H_0)$ and we have the bounds
    $$
    \norm{H_0^{1/2} W(\xi) (H_0+1)^{-1/2}} \leq C( \norm{\xi}_{H^1 \oplus L^2}),
    $$
    $$
    \norm{(N_1 +1)^{-1}(H_0+1)^{-1/2} W(\xi)(H_0+1)^{1/2}(N_1+1)} \leq C(\norm{\xi}_{H^1 \oplus L^2}).
    $$
    If $\xi,\eta \in L^2 \oplus L^2$, we have the bound
    $$
    \norm{(W(\xi) - W(\eta)) (N_1+N_2+1)^{-1/2} }\leq C \norm{\xi - \eta}_{L^2 \oplus L^2} (1+\norm{\xi}_{L^2 \oplus L^2} + \norm{\eta}_{L^2 \oplus L^2}).
    $$
\end{lem}

\begin{proof}
The proof is delegated to Appendix \ref{appendixlemma} to simplify the reading.
\end{proof}

\begin{lem} \label{N_1wigner}
     Let $(\rho_\eps)_{\eps \in (0,1]}$ be a family of normal states satisfying Assumptions \eqref{assumption1} and \eqref{assumption2}, be such that $\rho_{\eps_n} \longrightarrow \mu$. Then, along the same subsequence, we have, for every $\alpha \in \R$,
    $$
    e^{-i \alpha N_1} \rho_{\eps_n} e^{i \alpha N_1} \longrightarrow \mu.
    $$
    Furhtermore, the family $(e^{-i \alpha N_1} \rho_{\eps_n} e^{i \alpha N_1})_{\eps \in (0,1]}$ satisfies Assumptions \eqref{assumption1} and \eqref{assumption2} with the same constants.
\end{lem}

\begin{proof}
    A straightforward computation shows that, (remember that $W(\xi)$ and $N_1$ hide a semiclassical parameter) for every $\xi =(\xi_1,\xi_2) \in L^2 \oplus L^2$,
    $$
    e^{i \alpha N_1 } W(\xi) e^{-i \alpha N_1} = W( e^{i \eps \alpha } \xi_1, \xi_2 ).
    $$
    Therefore one can compute, using cyclicity of the trace,
    \begin{align*}
        \Tr(e^{-i \alpha N_1} \rho_{\eps} e^{i \alpha N_1} W(\xi)) &= \Tr(\rho_{\eps} W( e^{i \eps \alpha } \xi_1, \xi_2 )).
    \end{align*}
    It remains to show that
    $$
    \underset{n \longrightarrow \infty} \lim \Tr(\rho_{\eps_n} W( e^{i \eps_n \alpha } \xi_1, \xi_2 )) = \underset{n \longrightarrow \infty} \lim \Tr(\rho_{\eps_n} W( \xi_1, \xi_2 )) =  \int_{L^2 \oplus L^2}  e^{\sqrt{2} i \mathrm{Re} \langle z, \xi \rangle}    \di \mu(z).
    $$
    Indeed, by Lemma \ref{WxiQH0},
    \begin{align*}
        \abs{\Tr(\rho_{\eps} W( e^{i \eps \alpha } \xi_1, \xi_2 ))- \Tr(\rho_{\eps} W( \xi_1, \xi_2 ))} & \leq \Tr(\rho_\eps (N_1+1)) \norm{(W( e^{i \eps \alpha } \xi_1) - W(\xi_1))(N_1+1)^{-1/2}} \\
        &\leq  C (\norm{\xi}_{L^2 \oplus L^2})  \abs{e^{i \eps \alpha } -1} \underset{\eps \longrightarrow 0} \longrightarrow 0.
    \end{align*}
    The propagation of Assumptions \eqref{assumption1} and \eqref{assumption2} is straightforward since $e^{-i \alpha N_1}$ commutes with both $N_1^k$ and $H_0$.
\end{proof}

\begin{rem}
  Taking $\alpha = t \langle B_\infty, B_\infty + 2 f_\infty \rangle$, and noticing that $[N_1,\hat H_\infty ]=0$, one can replace the evolution given by $\hat H_\infty$ by the one given by $\hat H_\infty - \eps N_1 \langle B_\infty, B_\infty + 2 f_\infty \rangle$ and get the same Wigner measure without affecting the subsequence.
Therefore we will forget about the $\eps N_1 \langle B_\infty, B_\infty + 2 f_\infty \rangle$ term for the rest of this paper, and keep the notation $\hat H_\infty$ for $\hat H_\infty - \eps N_1 \langle B_\infty, B_\infty + 2 f_\infty \rangle$.
\end{rem}

\subsection{The Duhamel formula for the dressed Hamiltonian}
We recall the assumptions \eqref{assumption1}, \eqref{assumption2} on the initial states:
\begin{align*}
    &\exists C>0, \forall \eps \in (0,1], \forall k \in \N, \Tr[\rho_\eps N_1^k] \leq C^k. \\
    &\exists C>0, \forall \eps \in (0,1], \Tr[\rho_\eps H_0] \leq C.
\end{align*}
They ensure the existence of a Wigner measure for the family of normal states $(\rho_\eps)_{\eps \in (0,1]}$. This property will be conserved by the dynamics generated by $\hat H_\infty$. 
Let us first see how the first assumption translates into for the the family $(\rho_\eps)_{\eps \in (0,1]}$.
\begin{lem} \cite[Lemma 4.2]{ammari2017sima} \label{diagonalemma}
    The family of normal states $(\rho_\eps)_{\eps \in (0,1]}$ satisfies Assumption \eqref{assumption1} if and only if there exists a diagonalization of $\rho_\eps$ under the form
    $$
    \rho_\eps = \sum_{i \in \N} \lambda_i | \varphi_i \rangle \langle \varphi_i |,
    $$
    such that for all $i \in \N$, $\lambda_i \geq 0$, $\varphi_i \in \mathrm{Ran} (\Id_{[0, C]} (N_1))$. 
\end{lem}

\begin{rem}
    In other words, Assumptions \eqref{assumption1} translates the fact that we are considering the mean-field scaling $\eps \rightarrow 0,n \eps \sim 1$ and not the full semiclassical scaling $\eps \rightarrow 0$ on the Fock space.
\end{rem}

\begin{defn} Let $\rho_\eps$ be a normal state. We define the dressed evolution of $\rho_\eps$ to be
$$
\rho_\eps(t) = e^{-\frac{i t}{\eps}  \hat H_\infty} \rho_\eps e^{\frac{i t}{\eps} \hat H_\infty}.
$$
We also define the dressed evolution of $\rho_\eps$ in the interaction picture to be
$$
\tilde \rho_\eps(t) = e^{\frac{i t}{\eps} H_0} \rho_\eps (t) e^{-\frac{i t}{\eps}H_0} =e^{\frac{i t}{\eps} H_0} e^{-\frac{i t}{\eps} \hat H_\infty} \rho_\eps e^{\frac{i t}{\eps} \hat H_\infty} e^{-\frac{i t}{\eps }H_0}.$$
\end{defn}
The following proposition ensures the existence of a non empty set of Wigner measures for the evolved states.
\begin{prop} \label{propagation} Let $(\rho_\eps)_{\eps \in (0,1]}$ be a family of normal states satisfying \eqref{assumption1} and \eqref{assumption2}. There exists $C>0$ such that
$$
\begin{aligned} 
    &\forall t \in \R, \forall \eps \in (0,1], \forall k \in \N, \Tr[\rho_\eps (t) N_1^k] \leq C^k. \\
    &\forall t \in \R, \forall \eps \in (0,1], \Tr[\rho_\eps (t) H_0] \leq C.
\end{aligned}
$$
and the same holds for $\tilde \rho_\eps(t)$.
\end{prop}

\begin{proof}
The first statement is immediate as $[\hat H_\infty , N_1] = 0$. For the second one, we write, using Lemma \ref{diagonalemma}, and the fact that $H_0$ and $\hat H_\infty$ preserve $N_1$,
\begin{align*}
    \Tr(\rho_\eps(t) H_0) &= \sum_{i \in \N} \lambda_i \langle e^{-\frac{i t}{\eps} \hat H_\infty} \varphi_i, H_0 e^{-\frac{i t}{\eps} \hat H_\infty} \varphi_i \rangle \\
    &= \sum_{i \in \N}\lambda_i  \sum_{n = 0}^{\lfloor C\eps^{-1} \rfloor}  \langle e^{-\frac{i t}{\eps} \hat H_\infty^{(n)} } \varphi_i^{(n)}, H_0^{(n)} e^{-\frac{i t}{\eps} \hat H_\infty^{(n)}} \varphi_i^{(n)} \rangle.
\end{align*}
By Lemma \ref{klmnbound}, we have $ \pm (\hat H_\infty^{(n)} - H_0^{(n)}) = \pm \hat H_{\infty,I}^{(n)} \leq a H_0^{(n)}+b$ with $a<1$, frow which we deduce the bounds
\begin{align*}
H_{0}^{(n)} &\leq \frac{1}{1-a} \hat H_\infty^{(n)} + \frac{b}{1-a},\\
H_\infty^{(n)} &\leq (1+a) H_0^{(n)} +b.
\end{align*}
We now compute
\begin{align*}
     \langle e^{-\frac{i t}{\eps} \hat H_\infty^{(n)} } \varphi_i^{(n)}, H_0^{(n)} e^{-\frac{i t}{\eps} \hat H_\infty^{(n)}} \varphi_i^{(n)} \rangle &\leq \frac{1}{1-a} \left( \langle e^{-\frac{i t}{\eps} \hat H_\infty^{(n)} } \varphi_i^{(n)}, \hat H_\infty^{(n)} e^{-\frac{i t}{\eps} \hat H_\infty^{(n)}} \varphi_i^{(n)} \rangle + b \norm{ e^{-\frac{i t}{\eps} \hat H_\infty^{(n)}} \varphi_i^{(n)}  }^2  \right)\\
     &= \frac{1}{1-a} \langle \varphi_i^{(n)}, \hat H_\infty^{(n)} \varphi_i^{(n)} \rangle + \frac{b}{1-a} \norm{ \varphi_i^{(n)}  }^2 \\
     &\leq \frac{1}{1-a} \langle \varphi_i^{(n)}, ((1+a)H_0^{(n)}+b) \varphi_i^{(n)} \rangle + \frac{b}{1-a} \norm{ \varphi_i^{(n)}  }^2 \\
     &\leq C \langle \varphi_i^{(n)}, (H_0^{(n)}+1) \varphi_i^{(n)} \rangle.
\end{align*}
Summing in $n\leq \lfloor C \eps^{-1} \rfloor$ and in $i \in \N$ we can conclude that 
$$
\Tr(\rho_\eps(t) H_0) \leq  C \Tr(\rho_\eps (H_0+1)) \leq \tilde C.
$$  
Since $\Tr (\tilde \rho_\eps (t) N_1^k ) = \Tr (\rho_\eps(t) N_1^k) $ and $\Tr (\tilde \rho_\eps (t) H_0 ) = \Tr (\rho_\eps(t) H_0)$, the same  holds for $\tilde \rho_\eps (t)$.
\end{proof}

\begin{prop}
    Consider a family of normal states $(\rho_\eps)_{\eps \in (0,1]}$ satisfying the assumptions \eqref{assumption1} and \eqref{assumption2}. Then the following Duhamel formula holds:
    $\forall t \in \R, \forall \xi \in H^1 \oplus L^2$, $$
    \Tr(\tilde \rho_\eps(t) W(\xi)) = \Tr(\rho_\eps W(\xi))+ \frac{i}{\eps} \int_0^t \Tr \left(\rho_\eps(s) [\hat H_{\infty,I},W(\xi_s)]   \right) \di s,
    $$
    where $\hat H_{\infty,I} = \hat H_\infty-H_0$ and $W(\xi_s) = W(\phi_s^0 (\xi))= e^{-\frac{i s}{\eps}H_0} W(\xi) e^{\frac{i s}{\eps} H_0}$.
\end{prop}

\begin{proof}
    By Lemma \ref{diagonalemma} we can express the trace as
    $$
    \Tr (\tilde \rho_\eps(t) W(\xi)) = \sum_{j \in \N} \lambda_j \langle e^{\frac{i t}{\eps} H_0} e^{-\frac{i t}{\eps} \hat H_\infty} \varphi_j, W(\xi) e^{\frac{i t}{\eps} H_0} e^{-\frac{i t}{\eps} \hat H_\infty} \varphi_j \rangle.
    $$
    Since $Q(H_0) = Q(\hat H_\infty)$ and $\varphi_j \in Q(H_0)$ by \eqref{assumption2}, Proposition \ref{propagation} and Lemma \ref{WxiQH0} ensure that for any $j \in \N$, the map
    $$
    t \in \R \longmapsto  \langle e^{\frac{i t}{\eps} H_0} e^{-\frac{i t}{\eps} \hat H_\infty} \varphi_j, W(\xi) e^{\frac{i t}{\eps} H_0} e^{-\frac{i t}{\eps} \hat H_\infty} \varphi_j \rangle \in \R
    $$
    is differentiable and that
    \begin{align*}
        \frac{\di}{\di t} \langle e^{\frac{i t}{\eps} H_0} e^{-\frac{i t}{\eps} \hat H_\infty} \varphi_j, W(\xi) e^{\frac{i t}{\eps} H_0} e^{-\frac{i t}{\eps} \hat H_\infty} \varphi_j \rangle & = \frac{i}{\eps} \langle e^{\frac{i t}{\eps} H_0} (\hat H_\infty - H_0 )e^{-\frac{i t}{\eps} \hat H_\infty} \varphi_j, W(\xi) e^{\frac{i t}{\eps} H_0} e^{-\frac{i t}{\eps} \hat H_\infty} \varphi_j \rangle \\
        & - \frac{i}{\eps} \langle e^{\frac{i t}{\eps} H_0} e^{-\frac{i t}{\eps} \hat H_\infty} \varphi_j, W(\xi) e^{\frac{i t}{\eps} H_0} (\hat H_\infty - H_0 ) e^{-\frac{i t}{\eps} \hat H_\infty} \varphi_j \rangle \\
        &= \frac{i}{\eps} \langle e^{-\frac{i t}{\eps} \hat H_\infty} \varphi_j, [\hat H_{\infty,I}, e^{-\frac{i t}{\eps} H_0}W(\xi) e^{\frac{i t}{\eps} H_0} ]  e^{-\frac{i t}{\eps} \hat H_\infty} \varphi_j \rangle \\
        &= \frac{i}{\eps} \langle e^{-\frac{i t}{\eps} \hat H_\infty} \varphi_j, [\hat H_{\infty,I},W(\xi_t)]  e^{-\frac{i t}{\eps} \hat H_\infty} \varphi_j \rangle.
    \end{align*}
    Integrating between $0$ and $t \in \R$ we obtain 
    \begin{align*}
    \langle e^{\frac{i t}{\eps} H_0} e^{-\frac{i t}{\eps} \hat H_\infty} \varphi_j, W(\xi) e^{\frac{i t}{\eps} H_0} e^{-\frac{i t}{\eps} \hat H_\infty} \varphi_j \rangle &= \langle \varphi_j, W(\xi) \varphi_j \rangle \\ &+ \frac{i}{\eps} \int_0^t \langle e^{-\frac{i s}{\eps} \hat H_\infty} \varphi_j, [\hat H_{\infty,I},W(\xi_s)]  e^{-\frac{i s}{\eps} \hat H_\infty} \varphi_j \rangle \di s.
    \end{align*}
    The sum $\sum_j \lambda_j \langle \varphi_j, W(\xi) \varphi_j \rangle$ converges. Furthermore,
    \begin{align*}
     \abs{ \langle e^{-\frac{i s}{\eps} \hat H_\infty} \varphi_j, [\hat H_{\infty,I},W(\xi_s)]  e^{-\frac{i s}{\eps} \hat H_\infty} \varphi_j \rangle \di s}  &\leq 2 \norm{(H_0+1)^{1/2} e^{-\frac{i s}{\eps} \hat H_\infty} \varphi_j} \\
     & \ \ \ \times \norm{(H_0+1)^{-1/2} \hat H_{\infty,I} (H_0+1)^{-1/2}}
      \\
     & \ \ \ \times \norm{(H_0+1)^{1/2} W(\xi_s) e^{-\frac{i s}{\eps} \hat H_\infty} \varphi_j} \\
     &\leq C(\norm{\xi_s}_{H^1 \oplus L^2}) \norm{(H_0+1)^{1/2} e^{-\frac{i s}{\eps} \hat H_\infty} \varphi_j}^2 \\
     &= C(\norm{\xi}_{H^1 \oplus L^2}) \langle e^{-\frac{i s}{\eps} \hat H_\infty} \varphi_j, (H_0+1) e^{-\frac{i s}{\eps} \hat H_\infty} \varphi_j \rangle,
    \end{align*}
    where we have used Lemma \ref{WxiQH0} and the KLMN bound from Lemma \ref{klmnbound} to get
    $$
   \norm{(H_0+1)^{-1/2} \hat H_{\infty,I} (H_0+1)^{-1/2}} \leq a+C_a.
    $$
    Therefore, 
    $$
    \int_0^t \abs{ \lambda_j \frac{i}{\eps} \langle e^{-\frac{i s}{\eps} \hat H_\infty} \varphi_j, [\hat H_{\infty,I},W(\xi_s)]  e^{-\frac{i s}{\eps} \hat H_\infty} \varphi_j \rangle} \di s
    $$
    is summable via Proposition \ref{propagation} and Fubini's theorem yields the result. 
\end{proof}

\subsection{Analysis of the commutator}

\begin{lem} \label{commutator}
    In the sense of quadratic forms, we have the following expansion: for every $\varphi \in \mathrm{Ran} (\Id_{N_1 \leq C})$,
    $$
     \langle \varphi, \frac{i}{\eps} [\hat H_{\infty,I}, W(\xi_s)] \varphi \rangle = \langle \varphi , W(\xi_s) \sum_{j=0}^3 \eps^j \mathscr B_j (\xi_s) \varphi \rangle,
    $$
    for every $0 \leq j \leq 3$, and for every $\xi \in H^1 \oplus L^2$,Furthermore, there exists $C_j(\norm{\xi}_{H^1 \oplus L^2})>0$ such that for every $\varphi \in \mathrm{Ran} (\Id_{N_1 \leq C})$,
    $$
   \langle \varphi, (N_1+1)^{-1} (H_0+1)^{-1/2} \mathscr B_j (\xi_s) (H_0+1)^{-1/2} (N_1+1)^{-1} \varphi \rangle \leq C_j(\norm{\xi}_{H^1 \oplus L^2}) \norm{\varphi}^2.
    $$
\end{lem}

\begin{proof}
The key element is to recognize $\hat H_{\infty,I}$ as the Wick quantization of the interacting dressed classical Hamiltonian whose expression we recall is
    \begin{align*}
        \hat h_{\infty,I} (u,\alpha) &=  \int 2 \mathrm{Re} \langle \alpha, f_{\sigma_0} e^{-ik.x} \rangle \abs{u(x)}^2  \di x \\
        &+  \iint \abs{u(x)}^2 V_\infty(x-y) \abs{u(y)}^2 \di x \di y \\
        &+ \int (2 \mathrm{Re} \langle \alpha, k B_\infty e^{-ik.x} \rangle)^2 \abs{u(x)}^2 \di x  \\
        & -2 \int     \overline{u(x)} ( \langle \alpha, k B_\infty e^{-i k . x} \rangle  . D_{x} + D_x . \langle k B_\infty e^{-i k .x},\alpha \rangle )u(x) \di x. 
    \end{align*}
More precisely, from the definition of the domain of $D(\hat H_\infty)$ and the preservation of $N_1$ under the form $[\hat H_{\infty,I},N_1]=0$, we deduce that 
$$
\hat H_{\infty,I} = \Id_{N_1 \leq C} (\hat h_{\infty,I})^{\mathrm{Wick}} = (\hat h_{\infty,I})^{\mathrm{Wick}} \Id_{N_1 \leq C}.
$$
Hence, for every $\varphi \in \mathrm{Ran} (\Id_{N_1 \leq C})$,
\begin{align*}
    \langle \varphi, [\hat H_{\infty,I}, W(\xi_s)] \varphi \rangle &= \langle \varphi, (\Id_{N_1 \leq C} (\hat h_{\infty,I})^{\mathrm{Wick}} W(\xi_s) - W(\xi_s)  (\hat h_{\infty,I})^{\mathrm{Wick}} \Id_{N_1 \leq C} ) \varphi \rangle  \\
    &=  \langle \Id_{N_1 \leq C} \varphi,   (\hat h_{\infty,I})^{\mathrm{Wick}} W(\xi_s) \varphi \rangle - \langle \varphi , W(\xi_s) (\hat h_{\infty,I})^{\mathrm{Wick}}  \Id_{N_1 \leq C}  \varphi \rangle  \\
    &= \langle \varphi, [(\hat h_{\infty,I})^{\mathrm{Wick}} , W(\xi_s)] \varphi \rangle. 
\end{align*}
This means that when restricted to $\mathrm{Ran}(\Id_{N_1\leq C})$, $[\hat H_{\infty,I}, W(\xi_s)]$ and $[(\hat h_{\infty,I})^{\mathrm{Wick}}, W(\xi_s)]$ coincide in the sense of quadratic forms. Now, using commutation rules between Wick quantizations and Weyl operators, one has
\begin{align*}
[(\hat h_{\infty,I})^{\mathrm{Wick}}, W(\xi_s)] &= W(\xi_s) [W(\xi_s)^* (\hat h_{\infty,I})^{\mathrm{Wick}} W(\xi_s)- (\hat h_{\infty,I})^{\mathrm{Wick}}] \\ &= W(\xi_s) \left(  \hat h_{\infty,I} \left(\cdot + \frac{i \eps}{\sqrt{2}} \xi_s \right) - \hat h_{\infty,I}(\cdot)  \right)^{\mathrm{Wick}}.
\end{align*}
Since $\hat h_{\infty,I}$ is polynomial of order $4$, we can expand $\hat h_{\infty,I} \left(\cdot + \frac{i \eps}{\sqrt{2}} \xi_s \right) - \hat h_{\infty,I}(\cdot)$ into 
$$
\frac{\eps}{\sqrt{2}} \di \hat h_{\infty,I}(\cdot) (i \xi_s) + \mathcal O(\eps^2),
$$
where the remainder $\mathcal O(\eps)$ contains only polynomial terms of order two or less. More specifically, let's look at each term in $\hat h_{\infty,I}$ by their degree. The term 
$$
 \int 2 \mathrm{Re} \langle \alpha, f_{\sigma_0} e^{-ik.x} \rangle \abs{u(x)}^2  \di x + \int (2 \mathrm{Re} \langle \alpha, k B_\infty e^{-ik.x} \rangle)^2 \abs{u(x)}^2 \di x
$$
is polynomial of degree $2$ in $u$ and $2$ in $\alpha$. Commuting with the Weyl operator means losing $1$ degree, so each of the $\mathscr B_j (\xi_s)$ is a quantization of a polynomial of degree still less than $2$ in $u$ and $2$ in $\alpha$, which is controlled by $(N_1+1) (N_2+1)$. The term
$$
 \iint \abs{u(x)}^2 V_\infty(x-y) \abs{u(y)}^2 \di x \di y 
$$
is of degree $4$ in $u$, therefore the commutator is of degree $3$, hence its quantization is controlled by $(N_1+1)^{3/2}$. The last term
$$
-2 \int     \overline{u(x)} ( \langle \alpha, k B_\infty e^{-i k . x} \rangle  . D_{x} + D_x . \langle k B_\infty e^{-i k .x},\alpha \rangle )u(x) \di x
$$
is of degree $2$ in $\alpha$, $1$ in $u$ and $1$ in $\nabla u$, so its derivative is controlled by either $(N_2+1)^{1/2} (N_1+1)^{1/2} (\di \Gamma_1 (-\Delta)+1)^{1/2}$, $(N_2+1) (N_1+1)^{1/2} $ or $(N_2+1) (\di \Gamma_1(-\Delta)+1)^{1/2}$. Every of these terms is controlled by $(N_2+1)^2(H_0+1)$, which concludes.
\end{proof}

\begin{rem}
    The latter proof gives us an expression of $\mathscr B_0 (\xi_s)$ as the Wick quantization of a symbol $b_0( \xi_s)$:
    $$
    \mathscr B_0 (\xi_s) =\left( b_0(\xi_s) \right)^{\mathrm{Wick}} = \left( \frac{i}{\sqrt{2}} \di \hat h_{\infty,I}(\cdot)(i \xi_s) \right)^{\mathrm{Wick}}.
    $$
\end{rem}

\begin{cor} \label{highorderzeroes}
    We can then pass easily to the limit for the high order terms in the Duhamel formula: for any $t \in \R$, we get
    $$
    \sum_{j=1}^3 \eps^j \int_0^t \Tr (\rho_\eps(s) W(\xi_s) \mathscr B_j(\xi_s))\di s \underset{\eps \rightarrow 0}{\longrightarrow} 0.
    $$
\end{cor}

\begin{proof}
    Using Lemma \ref{WxiQH0} as well as the propagation of regularity of Proposition \ref{propagation}, we bound, for any $j=1,2,3$,
    \begin{align*}
        \abs{\Tr(\rho_\eps(s) W(\xi_s) \mathscr B_j(\xi_s))} &\leq C_j(\norm{\xi_s}_{H^1 \oplus L^2}) \Tr((N_1+1)^2(H_0+1) \rho_\eps (s)) \\
        &\leq C_j (\norm{\xi}_{H^1 \oplus L^2}),
    \end{align*}
    yielding the result.
\end{proof}

\begin{lem} \label{lemmeb0}
    Let $\xi \in H^1 \oplus L^2$. Then there exists a family of compact symbols
    $$
    b_0^m (\xi) \in \bigoplus_{\substack{0\leq p, q \leq 2,\\ 2 \leq p+q \leq 3 } } \mathcal P_{p,q}^\infty (L^2 \oplus L^2), m \in \N,
    $$
    such that
    $$
    \norm{ (N_1+1)^{-1} (H_0+1)^{-1/2} (b_0 (\xi) - b_0^m (\xi))^{\mathrm{Wick}} (H_0+1)^{-1/2} (N_1+1)^{-1}  } \leq C (m, \norm{\xi}_{H^1 \oplus L^2}),
    $$
    with
    $$
     C (m, \norm{\xi}_{H^1 \oplus L^2}) \underset{ m \longrightarrow \infty} \longrightarrow 0.
    $$
\end{lem}

The proof is purely technical and is relegated to Appendix \ref{appendixlemma} to avoid complicating the text.

\subsection{Existence of the Wigner measures}
Proposition \ref{propagation} provides for every $t \in \R$ the existence of $\mu_t, \tilde \mu_t$ and an extraction $\eps_{k(t)} \rightarrow 0$ such that along the subsequence,
    $$
    \rho_{\eps_{k(t)}}(t) \longrightarrow \mu_t, \ \ \ \tilde \rho_{\eps_{k(t)}} (t) \longrightarrow \tilde \mu_t.
    $$
    Let us first look at properties of these Wigner measures, inherited from the bounds of Assumptions \eqref{assumption1} and \eqref{assumption2}.
    \begin{lem} \label{boundmeasure}Let $(\rho_\eps)_{\eps \in (0,1]}$ be a family of normal states satisfying Assumptions \eqref{assumption1} and \eqref{assumption2}, and $\eps_k \rightarrow 0$, and let $\mu \in \mathscr M((\rho_{\eps \in (0,1]}))$. Then, 
    $$
 \int_{L^2 \oplus L^2} \norm{u}_{L^2}^{2 k} \di \mu(u,\alpha) \leq C^k, \ \forall k \in \N,
 $$ 
        $$
 \int_{L^2 \oplus L^2} (\norm{\nabla u}^2_{L^2} + \norm{\alpha}^2_{L^2}) \di  \mu(u,\alpha) \leq C,        $$
 where $C$ is the constant appearing in \eqref{assumption1} and \eqref{assumption2}.
    \end{lem}
    \begin{proof}
        The proof can be found in \cite[Lemma 3.13]{ammari2011} or \cite[Lemma 4.12]{ammari2017sima}.
    \end{proof}
    \begin{cor} \label{supportmeasure} Let $(\rho_\eps)_{\eps \in (0,1]}$ be a family of normal states satisfying Assumptions \eqref{assumption1} and \eqref{assumption2}, and let $\mu \in \mathscr M((\rho_{\eps \in (0,1]}))$. Then the support of $\mu$ is included in 
    $$
    \{ (u,\alpha) \in H^1 \oplus L^2 | \norm{u}_{L^2}^2 \leq C \} = (H^1 \oplus L^2) \cap \{ (u,\alpha) \in L^2 \oplus L^2  | \norm{u}_{L^2}^2 \leq C \}.
    $$
    \end{cor}
\begin{proof}
    First, we compute
    \begin{align*}
        \mu ((L^2 \oplus L^2) \setminus (H^1 \oplus L^2) ) &= \mu (\{ \norm{\nabla u}_{L^2}^2 = \infty \} ) \\&= \mu \left( \bigcap_{p \in \N}  \{  \norm{\nabla u}_{L^2}^2 \geq p   \} \right) \\&= \underset{p \rightarrow \infty} \lim \mu ( \{  \norm{\nabla u}_{L^2}^2 \geq p   \}) \\ & \leq \underset{p \rightarrow \infty} \lim  \frac{1}{p} \int_{L^2 \oplus L^2} \norm{\nabla u}^2_{L^2} \di \mu(u,\alpha) \\
        &\leq \underset{p \rightarrow \infty} \lim \frac{C}{p} = 0,
    \end{align*}
    so that $\mu(H^1 \oplus L^2) = 1$. Similarly, for every $p,k \in \N^*$,
    \begin{align*}
        \mu \left(\left\{\norm{u}^2_{L^2} \geq C + \frac{1}{p} \right\} \right) \leq \left( C + \frac{1}{p} \right)^{-k} C^k \underset{k \rightarrow \infty} \longrightarrow 0,
    \end{align*}
    so that 
    $$
    \mu \left(\left\{\norm{u}^2_{L^2} >C \right\} \right) =\underset{p \rightarrow \infty} \lim \mu \left(\left\{\norm{u}^2_{L^2} \geq C + \frac{1}{p} \right\} \right)= 0,
    $$
    which concludes.
\end{proof}

\begin{rem}
    This corollary allow us to gain an interpretation concerning the necessity of the infrared cutoff $\sigma_0$. To be fixed, it is needed that $n \eps$ is uniformly bounded, which at the classical level corresponds to $\norm{u}_{L^2}^2$ being uniformly bounded. Let's look at the implications on the ground state energy. Without any restriction, the classical Hamiltonian energy is bounded from below:
    $$
    \underset{(u,\alpha) \in H^1 \oplus L^2} \inf h_\infty(u,\alpha) = - \infty. 
    $$
    Indeed, the interaction term $\int A_\alpha \abs{u}^2$ is cubic and has no sign, therefore its behavior dominates the free quadratic part when both $u$ and $\alpha$ are large (one can easily check this by computing $h(\lambda u,\lambda \alpha)$ for $\lambda \in \R$, and letting $\lambda \longrightarrow \pm \infty$). However, with $\sigma_0$ fixed, the support of the Wigner measure $\mu$ cannot be too large. Virtually, the analysis is only made when $(u,\alpha) \in \mathrm{supp} \ \mu$, on which subset the ground state energy now becomes bounded from below:
    $$
    \underset{(u,\alpha) \in \mathrm{supp} \ \mu} \inf h_\infty(u,\alpha) > - \infty.
    $$
    To see this, it suffices to compute
    \begin{align*}
        \int A_\alpha \abs{u}^2 &\gtrsim - \norm{\alpha}^2_{L^2} \norm{\nabla u}^{1/2}_{L^2} \norm{u}_{L^2}^{3/2} \\
        &\gtrsim  - \norm{\alpha}^2_{L^2} \norm{\nabla u}^{1/2}_{L^2},
    \end{align*}
    which is now controlled on $\mathrm{supp} \ \mu$ by the free Hamiltonian $\norm{\nabla u}^2_{L^2}+\norm{\alpha}^2_{L^2}$.
\end{rem}
    Let us go back to the Duhamel formula. To be able to take the limit in the latter, we need a common subsequence of extraction $\eps_k \rightarrow 0$ independent of the time. The following allows to prove that such a subsequence can be constructed.
\begin{prop} \label{wignerexistence}
    Let $(\rho_\eps)_{\eps \in (0,1]}$ be a family of normal states satisfying Assumptions \eqref{assumption1} and \eqref{assumption2}, and $\eps_k \rightarrow 0$. There exists an extraction $(\eps_{\varphi(k)})_{k \in \N}$, and for every $t \in \R$ a probability measure $\mu_t$ on $L^2 \oplus L^2$ such that 
    \begin{equation} \label{wignerexistence1}
        \rho_{\eps_{\varphi(k)}}(t) \longrightarrow \mu_t,
    \end{equation}
    \begin{equation} \label{wignerexistence2}
      \tilde \rho_{\eps_{\varphi(k)}}(t) \longrightarrow \tilde \mu_t = (\phi^0_{-t})_\# \mu_t,
    \end{equation}
    \begin{equation} \label{wignerexistence3}
        \forall \xi \in L^2 \oplus L^2, \tilde \rho_{\eps_{\varphi(k)}}(t) W(\xi_t) \longrightarrow e^{\sqrt{2} i \mathrm{Re} \langle \xi_t, z \rangle } \di \tilde \mu_t(z).
    \end{equation}
\end{prop}

\begin{proof}
    \textbf{Step 1.} $G_\eps (t,\xi) :=  \Tr ( \tilde \rho_{\eps_k}(t) W(\xi))$ is uniformly equicontinuous on bounded subsets of $\R \times (H^1 \oplus L^2)$. \\
    Firstly, for $t \in \R, \xi,\eta \in H^1 \oplus L^2$, by Lemma \ref{WxiQH0},
    \begin{align*}
        \abs{G_\eps (t,\xi) - G_\eps (t,\eta)} & \leq  \norm{(W(\xi) - W(\eta)) (N_1+1)^{-2} (H_0+1)^{-1} } \Tr(\tilde \rho_\eps (t) (N_1+1)^2 (H_0+1) ) \\
        &\leq   \norm{(W(\xi) - W(\eta)) (N_1+N_2+1)^{-1/2} } \Tr(\tilde \rho_\eps (t) (N_1+1)^2 (H_0+1) ) \\
        &\leq C \norm{\xi - \eta}_{L^2 \oplus L^2} (1+\norm{\xi}_{L^2 \oplus L^2} + \norm{\eta}_{L^2 \oplus L^2}) \\
        &\leq C \norm{\xi - \eta}_{H^1 \oplus L^2} (1+\norm{\xi}_{H^1 \oplus L^2} + \norm{\eta}_{H^1 \oplus L^2}).
    \end{align*}
    Secondly, for $t_1 \leq t_2 \in \R, \xi \in H^1 \oplus L^2$, by Lemma \ref{commutator},
    \begin{align*}
        \abs{G_\eps (t_2,\xi) - G_\eps(t_1, \xi)} &\leq \int_{t_1}^{t_2} \abs{ \Tr \left(  \rho_\eps (s) \frac{i}{\eps} [\hat H_{I,\infty},W(\xi_s)] \right) \di s} \\ &\leq \int_{t_1}^{t_2} \norm{(N_1+1)^{-1} (H_0+1)^{-1/2}  \sum_{j=0}^3 \eps^k \mathscr B_j(\xi_s)  (H_0+1)^{-1/2} (N_1+1)^{-1} } \\
        & \ \ \ \ \  \times \norm{ (N_1+1)^{-1}   (H_0+1)^{-1/2} W(\xi_s)  (N_1+1)  (H_0+1)^{1/2}} \\
        & \ \ \ \ \ \times \Tr( \rho_\eps (s) (N_1+1)^2 (H_0+1)) \di s \\
        &\leq \sum_{j=0}^3 C_j (\norm{\xi}_{H^1 \oplus L^2}) \abs{t_1-t_2} \\
        & \leq C(\norm{\xi}_{H^1 \oplus L^2}) \abs{t_1-t_2},
    \end{align*}
    which concludes.\\
    \textbf{Step 2.} Diagonal extraction argument for $t \mapsto \tilde \mu_t$. \\
    Let $(t_j)_{j \in \N}$ be a dense countable family in $\R$. Let $\varphi_j : \N \rightarrow \N$ be an extraction satisfying $$\tilde{\rho}_{\eps_{\varphi_j(k)}}(t_j) \longrightarrow \tilde \mu_{t_j}.$$
We define a new extraction by $ \varphi : k \in \N \mapsto \varphi_0 \circ \varphi_1 \circ ... \circ \varphi_k (k) \in \N$. Along this subsequence we have that
$$\forall j \in \N, \tilde{\rho}_{\eps_{\varphi(k)}}(t_j) \longrightarrow \tilde \mu_{t_j},$$
or in other terms,
$$
\forall j \in \N, \forall \xi \in H^1 \oplus L^2, \underset{k \rightarrow + \infty} \lim G_{\eps_{\varphi(k)}} (t_j , \xi) = \int_{L^2 \oplus L^2} e^{\sqrt{2} i \mathrm{Re} \langle \xi, z \rangle} \di \tilde \mu_{t_j} (z) =: G_0 (t_j, \xi).
$$
We define for every $t \in \R$,
$$
G_0 (t, \xi) = \underset{t_j \rightarrow t} \lim G_0(t_j, \xi).
$$
It is well defined using the uniform equicontinuity of \textbf{Step 1.} We want to recognize $G_0(t,\cdot)$ as the Fourier transform of a probability measure. We have immediatly by taking the limit of $G_{\eps_{\varphi(k)}}$ and using uniform equicontinuity that $G_0(t,0) = 1$, $G(t,\cdot)$ is continuous and that $G_0(t,\cdot)$ is of positive type in the sense that $
\forall n \in \N, \forall (\alpha_j)_{1\leq j \leq n} \in \C^n, \forall (\xi_j)_{1\leq j \leq n} \in (H^1 \oplus L^2)^n,$
$$\sum_{1\leq i, j \leq n} \alpha_i \overline \alpha_j G_0(t, \xi_i - \xi_j) \geq 0.
$$
Therefore $G_0$ is the Fourier transform of a cylindrical measure $\mathfrak m_t$. On the other side, one checks that if $t_{j(k)} \rightarrow t$, the sequence of probability measures $(\mu_{t_{j(k)}})_k \subset \mathcal P (H^1 \oplus L^2)$ is tight. Indeed, Proposition \ref{propagation} and Lemma \ref{boundmeasure} ensures the existence of $C>0$ such that for every $k \in \N$,
\begin{equation}
    \int (\norm{\nabla u}^2_{L^2} + \norm{\alpha}^2_{L^2}) \di \tilde \mu_{t_{j(k)}}(u,\alpha) \leq C.
\end{equation}
Therefore there exists a weak narrow limit $\tilde \mu_t \in \mathcal P (H^1 \oplus L^2)$ of $\tilde \mu_{t_j}$. In fact, one has $\mathfrak m_t = \tilde \mu_t$. Indeed,
$$ \int e^{\sqrt{2} i \mathrm{Re} \langle \xi, z \rangle}  \di \tilde \mu_t (z) = \underset{t_j \rightarrow t} \lim \int e^{\sqrt{2} i \mathrm{Re} \langle \xi, z \rangle}  \di \tilde \mu_{t_j} (z)= \underset{t_j \rightarrow t} \lim  G_0(t_j,\xi) = G_0 (t,\xi) =  \int e^{\sqrt{2} i \mathrm{Re} \langle \xi, z \rangle}  \di \mathfrak m_t (z). $$
We finally have to check that
$$
\underset{k \rightarrow +\infty} \lim G_{\eps_{\varphi(k)}} (t,\xi) = G_0(t,\xi).
$$
We compute
\begin{align*}
    \abs{G_{\eps_{\varphi(k)}} (t,\xi) - G_0(t,\xi)} &\leq \abs{G_{\eps_{\varphi(k)}} (t,\xi)-G_{\eps_{\varphi(k)}} (t_j,\xi)} \\
    &+ \abs{G_{\eps_{\varphi(k)}} (t_j,\xi) - \int_{L^2 \oplus L^2} e^{\sqrt{2}i \mathrm{Re} \langle \xi, z \rangle} \di \tilde \mu_{t_j} (z)} \\
    &+ \abs{\int_{L^2 \oplus L^2} e^{\sqrt{2}i \mathrm{Re} \langle \xi, z \rangle} \di \tilde \mu_{t_j} (z) - \int_{L^2 \oplus L^2} e^{\sqrt{2}i \mathrm{Re} \langle \xi, z \rangle} \di \tilde \mu_{t} (z)}.
\end{align*}
Therefore,
\begin{align*}
\underset{k \rightarrow +\infty} \limsup \abs{G_{\eps_{\varphi(k)}} (t,\xi) - G_0(t,\xi)}  &\leq C(\norm{\xi}_{H^1 \oplus L^2}) \abs{t-t_j} \\ & + \abs{\int e^{\sqrt{2}i \mathrm{Re} \langle \xi, z \rangle} \di \tilde \mu_{t_j} (z) - \int e^{\sqrt{2}i \mathrm{Re} \langle \xi, z \rangle} \di \tilde \mu_{t} (z)}.
\end{align*}
Taking the limit $t_j \rightarrow t$ yields the result.
\\
    \textbf{Step 3.} Linking $\mu_t$ and $\tilde \mu_t$. \\
    By a similar argument we can assume that the subsequence $\eps_{\varphi(k)}$ also works for the convergence $\rho_{\eps(k)}(t) \rightarrow \mu_t$. A computation shows that
    $$
    \Tr(\tilde \rho_{\eps_\varphi(k)}(t) W(\xi)) = \Tr( \rho_{\eps_{\varphi(k)}} (t) W(\xi_t)) = \Tr( \rho_{\eps_{\varphi(k)}} (t) W(\phi^0_t (\xi))).
    $$
    Therefore, by taking the limit, we have 
    $$
    \int_{L^2 \oplus L^2} e^{\sqrt{2} i \mathrm{Re} \langle \xi ,z \rangle} \di \tilde \mu_t (z) = \int_{L^2 \oplus L^2} e^{\sqrt{2} i \mathrm{Re} \langle \xi ,\phi_{-t}^0(z) \rangle} \di \mu_t (z)=\int_{L^2 \oplus L^2} e^{\sqrt{2} i \mathrm{Re} \langle \xi ,z \rangle} \di (\phi_{-t}^0)_\# \mu_t (z).
    $$
    The Fourier transform of the two measures $\tilde \mu_t$ and $(\phi_{-t}^0)_\# \mu_t$ coincide, hence they must be equal.
    \textbf{Step 4.} Proof of \eqref{wignerexistence3}. \\
    For all $\xi,\eta \in L^2 \oplus L^2$, using the Weyl commutation relations,
    $$
    \Tr(\tilde \rho_{\eps_{\varphi(k)}} (t) W(\xi_t) W(\eta)) = \Tr (\tilde \rho_{\eps_{\varphi(k)}} (t) W(\xi_t +\eta)) e^{- i \frac{ \eps_{\varphi(k)}}{2} \mathrm{Im}\langle \xi_t, \eta \rangle}.
    $$
    Tkaing the limit $k \longrightarrow + \infty$, it holds 
    $$
    \underset{k \rightarrow + \infty} \lim \Tr (\tilde \rho_{\eps_{\varphi(k)}} (t) W(\xi_t) W(\eta)) = \int e^{\sqrt{2}i \mathrm{Re} \langle \xi_t + \eta,z \rangle} \di \tilde \mu_t (z) = \int  e^{\sqrt{2}i \mathrm{Re} \langle  \eta,z \rangle} e^{\sqrt{2}i \mathrm{Re} \langle  \xi_t,z \rangle}\di \tilde \mu_t (z),
    $$
    proving \eqref{wignerexistence3}.
\end{proof}

\subsection{Limit in the Duhamel formula}
Let us recall a useful result from \cite[Lemma 4.15]{ammari2017sima}. It is an adaptation of a result from \cite{ammari2008ahp} concerning the semiclassical limit of Wick quantizations of compact symbols.
\begin{prop}
    Consider a family of normal states $(\rho_{\eps})_{\eps \in (0,1]}$ such that $\rho_{\eps_j} \rightarrow \mu \in \mathcal P(L^2 \oplus L^2)$ and that for some $\delta>0$,
    $$
    \underset{j \in \N}{\sup} \norm{(N_1+N_2)^{\delta/2} \rho_{\eps_j} (N_1+N_2)^{\delta/2} } < + \infty.
    $$
    Then for all $b \in \bigoplus_{p+q < 2 \delta} \mathcal P_{p,q}^\infty (L^2 \oplus L^2)$,
    $$
    \underset{j \rightarrow \infty} \lim \Tr \left(\rho_{\eps_j} b^{\mathrm{Wick}} \right) = \int_{L^2 \oplus L^2} b(z) \di \mu (z).
    $$
\end{prop}

\begin{prop} \label{lemmacompact}
    Consider a family of normal states $(\rho_{\eps})_{\eps \in (0,1]}$ satisfying Assumptions \eqref{assumption1} and \eqref{assumption2}. Let $t\mapsto \tilde \mu_t$ be the weakly narrow continuous map of probability measures obtained from Proposition \ref{wignerexistence}. Then $t\mapsto \tilde \mu_t$ follows the following transport equation: for all $\xi \in L^2 \oplus L^2$,
    $$
    \int_{L^2 \oplus L^2} e^{\sqrt{2} i \mathrm{Re} \langle \xi,z \rangle} \di \tilde \mu_t(z) = \int_{L^2 \oplus L^2} e^{\sqrt{2} i \mathrm{Re} \langle \xi,z \rangle} \di \mu_0(z) + \int_0^t \int_{L^2 \oplus L^2} b_0(\xi_s)(z) e^{\sqrt{2} i \mathrm {Re} \langle \xi_s, z \rangle} \di  \mu_s (z) \di s, 
    $$
    with $b_0(\xi_s)(z) e^{\sqrt{2} i \mathrm {Re} \langle \xi_s, z \rangle} \in L^\infty_t ( \R, L^1 (L^2 \oplus L^2 , \mu_t))$.
\end{prop}

\begin{proof}
We want to take the semiclassical limit of the Duhamel formula, which we recall is
$$
\Tr (\tilde \rho_\eps (t) W(\xi)) = \Tr (\rho_\eps W(\xi)) + \int_0^t \Tr \left( \rho_\eps (s) W(\xi_s) \sum_{j=0}^3 \eps^j \mathscr B_j (\xi_s) \right) \di s,
$$
along the uniform subsequence. From Corollary \ref{highorderzeroes} and Proposition \ref{wignerexistence}, it only remains to show that
$$
\underset{k \rightarrow \infty} \lim \int_0^t \Tr \left(  \rho_{\eps_{\varphi(k)}} (s) W(\xi_s) \mathscr B_0(\xi_s) \right) \di s = \int_0^t \int_{L^2 \oplus L^2} b_0(\xi_s)(z) e^{\sqrt{2} i \mathrm {Re} \langle \xi_s, z \rangle} \di \mu_s (z) \di s.
$$
    This is where we need to approximate the symbol $b_0$ of the operator $\mathscr B_0$. Consider $\mathscr B_0^m (\xi_s) := (b_0^m(\xi_s))^{\mathrm{Wick}}$ constructed from Lemma \ref{lemmeb0}. Since $b_0^m$ is a compact symbol of order $3$ or less, and that the state $  \rho_{\eps_{\varphi(k)}} W(\xi_s)$ satisfies Assumptions \eqref{assumption1} and \eqref{assumption2}, we can use Lemma \ref{lemmacompact} to state that for every $s \in [0,t]$,
    $$
    \underset{k \rightarrow \infty} \lim \Tr \left(  \rho_{\eps_{\varphi(k)}} (s) W(\xi_s) \mathscr B_0^m (\xi_s) \right) = \int_{L^2 \oplus L^2} b_0^m (\xi_s) (z)  e^{\sqrt{2} i \mathrm {Re} \langle \xi_s, z \rangle} \di  \mu_s (z). 
    $$
We will now invoke the dominated convergence theorem to pass to the limit under the time integral. For the domination, we first express $ \rho_\eps (s) W(\xi_s)$ again, using the diagonal form of $\rho_\eps(s)$,
$$
 \rho_{\eps}(s) = \sum_{i \in \N}  \lambda^i_\eps(s) |\varphi_\eps^i (s) \rangle \langle \varphi^i_\eps(s) |,
$$
with $\varphi_\eps^i (s) \in \mathrm{Ran} \Id_{[0,C]} (N_1)$. Now, for every $i \in \N,$ using Lemma $\ref{dominationB0m}$,
\begin{align*}
    &\abs{\langle \varphi^i_\eps (s), W(\xi_s)  \mathscr B_0^m (\xi_s) \varphi^i_\eps (s) \rangle} \\
    &\leq \sum_{n=0}^{\lfloor C \eps^{-1} \rfloor} \abs{\langle (N_1+1)  \varphi_\eps^{i,n}(s) ,  (N_1+1)^{-1} W(\xi_s) \mathscr B_0^m (\xi_s) (N_1+1)^{-1} (N_1+1) \varphi_\eps^{i,n}(s)  \rangle} \\
    &\leq  \sum_{n=0}^{\lfloor C \eps^{-1} \rfloor} (n\eps +1)^2 \abs{  \langle (H_0+1)^{1/2} \varphi_\eps^{i,n}(s) ,   (H_0+1)^{1/2} \varphi_\eps^{i,n}(s)  \rangle}
    \\
    & \ \ \ \ \ \ \ \ \ \ \times \norm{   (N_1+1)^{-1} (H_0+1)^{-1/2}   W(\xi_s) (H_0+1)^{1/2}(N_1+1)}
    \\
    & \ \ \ \ \ \ \ \ \ \  \times\norm{(H_0+1)^{-1/2} (N_1+1)^{-1} \mathscr B_0^m (\xi_s) (N_1+1)^{-1} (H_0+1)^{-1/2}} \\
    &\leq (C+1)^2 C(\norm{\xi_s}_{H^1 \oplus L^2})C_m \norm{\xi_s}_{H^1 \oplus L^2} \sum_{n=0}^{\lfloor C \eps^{-1} \rfloor}  \langle \varphi_\eps^{i,n}(s) ,   (H_0+1) \varphi_\eps^{i,n}(s)  \rangle \\
    &=  C(m ,\norm{\xi}_{H^1 \oplus L^2}) \langle \varphi_\eps^{i}(s) ,   (H_0+1) \varphi_\eps^{i}(s)  \rangle,
\end{align*}
from which we deduce, using Proposition \ref{propagation} the domination
$$
\abs{\Tr( \rho_\eps W(\xi_s) \mathscr B_0^m(\xi_s))} \leq  C(m ,\norm{\xi}_{H^1 \oplus L^2}) \in L^1 ([0,t]).
$$
Hence we have 
$$
\underset{k \rightarrow \infty} \lim \int_0^t \Tr \left(  \rho_{\eps_{\varphi(k)}} (s) W(\xi_s) \mathscr B_{0}^m(\xi_s) \right) \di s = \int_0^t \int_{L^2 \oplus L^2} b_{0}^m(\xi_s)(z) e^{\sqrt{2} i \mathrm {Re} \langle \xi_s, z \rangle} \di  \mu_s (z) \di s.
$$
To get rid of the $m$ dependance, we compute
\begin{align*}
    &\abs{ \int_0^t \Tr \left(  \rho_{\eps_{\varphi(k)}} (s) W(\xi_s) \mathscr B_0(\xi_s) \right) \di s - \int_0^t \int_{L^2 \oplus L^2} b_0(\xi_s)(z) e^{\sqrt{2} i \mathrm {Re} \langle \xi_s, z \rangle} \di  \mu_s (z) \di s} \\
    &\leq \abs{ \int_0^t \Tr \left(  \rho_{\eps_{\varphi(k)}} (s) W(\xi_s) (\mathscr B_0(\xi_s)-\mathscr B_0^m (\xi_s)) \right) \di s } \\
    &+ \abs{\int_0^t \Tr \left(  \rho_{\eps_{\varphi(k)}} (s) W(\xi_s) \mathscr B_{0}^m(\xi_s) \right) \di s - \int_0^t \int_{L^2 \oplus L^2} b_{0}^m(\xi_s)(z) e^{\sqrt{2} i \mathrm {Re} \langle \xi_s, z \rangle} \di  \mu_s (z) \di s} \\
    &+ \abs{\int_0^t \int_{L^2 \oplus L^2} (b_{0}^m(\xi_s)(z)-b_0(\xi_s)(z) e^{\sqrt{2} i \mathrm {Re} \langle \xi_s, z \rangle} \di  \mu_s (z) \di s}.
\end{align*}
We know that for every $m$ the middle term can be made as small as we want as $\eps$ goes to zero, meaning we have margin to choose $m$. The smallness of the last term is ensured by Lemma \ref{lemmeb0m}. For the first term, doing the same proof as for the domination and using Lemma \ref{lemmeb0}, we conclude that 
$$
\abs{\Tr \left(  \rho_{\eps_{\varphi(k)}} (s) W(\xi_s) \mathscr B_0(\xi_s) \right)} \leq C(m,\norm{\xi}_{H^1 \oplus L^2})
$$
and the right hand side goes to zero as $m \rightarrow \infty$ in $L^1([0,t])$. Therefore, by first choosing $m$ large enough so that the first and last terms are small and then choosing $\eps$ small enough along the subsequence, we deduce the Proposition. 
\end{proof}

\begin{lem} \label{integralformulaX}
     Consider a family of normal states $(\rho_{\eps})_{\eps \in (0,1]}$ satisfying Assumptions \eqref{assumption1} and \eqref{assumption2}. Let $t\mapsto \tilde \mu_t$ be the weakly narrow continuous map of probability measures obtained from Proposition \ref{wignerexistence}. Then $t\mapsto \tilde \mu_t$ follows the following transport equation: for all $\xi \in L^2 \oplus L^2$,
    \begin{align*}
    \int_{L^2 \oplus L^2} e^{\sqrt{2} i \mathrm{Re} \langle \xi,z \rangle} \di \tilde \mu_t(z) &= \int_{L^2 \oplus L^2} e^{\sqrt{2} i \mathrm{Re} \langle \xi,z \rangle} \di \mu_0(z) \\ &+ \sqrt{2} i \int_0^t \int_{L^2 \oplus L^2} \mathrm{Re}\langle X(s,z),\xi \rangle e^{\sqrt{2} i \mathrm {Re} \langle \xi, z \rangle} \di \tilde \mu_s (z) \di s, 
    \end{align*}
    with $X(s,\cdot) = - i \phi_{-s}^0 \circ \nabla_{\overline{z}} \hat h_{\infty,I} \circ \phi^0_s $.
\end{lem}

\begin{proof}
    The expression for $b_0(\xi_s) (z)$ is
    \begin{align*}
        b_0(\xi_s) (z) &= \frac{i}{\sqrt{2}} \di \hat h_{\infty,I}(z) (i\xi_s) \\
        &= \frac{i}{\sqrt{2}} 2 \mathrm{Re} \langle \nabla_{\overline{z}} \hat h_{\infty,I} (z), i \phi_s^0 (\xi) \rangle \\
        &= \sqrt{2} i \mathrm{Re} \langle -i  \phi^0_{-s} \circ   \nabla_{\overline{z}} \hat h_{\infty,I} (z), \xi \rangle \\
        &= \sqrt{2} i \mathrm{Re} \langle  X(s,\phi_{-s}^0(z)), \xi \rangle.
    \end{align*}
    Therefore by Proposition \ref{wignerexistence},
    \begin{align*}
         \int_0^t \int_{L^2 \oplus L^2} b_0(\xi_s)(z) e^{\sqrt{2} i \mathrm {Re} \langle \xi_s, z \rangle} &\di \tilde \mu_s (z) \di s \\ &= \sqrt{2} i \int_0^t \int_{L^2 \oplus L^2} \mathrm{Re} \langle X(s,\phi^0_{-s}(z)),\xi \rangle e^{\sqrt{2} i \mathrm{Re} \langle \xi, \phi^0_{-s}(z) \rangle} \di \mu_s (z) \di s \\
         &=\sqrt{2} i \int_0^t \int_{L^2 \oplus L^2} \mathrm{Re}\langle X(s,z),\xi \rangle e^{\sqrt{2} i \mathrm {Re} \langle \xi, z \rangle} \di ( \phi^0_{-s})_\# \mu_s (z) \di s \\
         &=\sqrt{2} i \int_0^t \int_{L^2 \oplus L^2} \mathrm{Re}\langle X(s,z),\xi \rangle e^{\sqrt{2} i \mathrm {Re} \langle \xi, z \rangle} \di \tilde \mu_s (z) \di s
    \end{align*}
    \end{proof}
Writing the transport equation in the latter form allow us to use a powerful machinery for solving an abstract Liouville equation, which is shortly summarized in Appendix \ref{appendixliouville}.
\begin{prop} \label{propdresseddynamics}
    Consider a family of normal states $(\rho_\eps)_{\eps \in (0,1]}$ satisfying Assumptions \eqref{assumption1} and \eqref{assumption2}. Assume that
    $$
    \mathscr M (\rho_\eps , \eps \in (0,1]) = \{ \mu_0 \} \subset \mathcal P (L^2 \oplus L^2).
    $$
    Then, 
    $$
    \mathscr M (e^{- \frac{i t}{\eps} \hat H_\infty} \rho_\eps e^{ \frac{i t}{\eps} \hat H_\infty} , \eps \in (0,1]) = \{(\hat \phi_t)_\# \mu_0 \} \subset \mathcal P (L^2 \oplus L^2),
    $$
    where $\hat \phi_t$ is the classical Hamiltonian flow of the dressed evolution.
\end{prop}

\begin{proof}
    The only thing that remains to show is that $\tilde \mu_t = (\hat \phi_t)_\# \mu_0$. We make use of a probabilistic approach as presented in Appendix \ref{appendixliouville}. Let us show that $t \mapsto \tilde \mu_t$ satisfies the hypothesis of Proposition \ref{globalsuperposition}. 

It is clear from Lemma \ref{boundmeasure} that $X$ satisfies the following bound 
$$
t \mapsto \int_{L^2 \oplus L^2} \norm{X(t,z)}_{L^2 \oplus L^2} \di \tilde \mu_t (z) \di t \in L^1_{\mathrm{loc}}(\R).
$$

Proposition \ref{lemmacompact} and Lemma \ref{liouvilleequivalence} ensures that $t \mapsto \tilde \mu_t$ satisfies the following transport equation:

$\forall \varphi \in  \mathscr C^\infty_{0,\text{cyl}} (\R \times (L^2 \oplus L^2))$,
    $$
    \int_0^t \int_{L^2 \oplus L^2} \left(  \partial_t \varphi (t,z) + \mathrm{Re} \langle X(t,z) ,  \nabla_{L^2 \oplus L^2} \varphi (t,z)   \rangle_{ L^2 \oplus L^2}   \right) \di \tilde \mu_t (z) \di t = 0.
    $$

    Then there exists a measurable flow $\hat \Phi_t$, associated to the initial value problem, which satisfies $\tilde \mu_t = (\hat \Phi_t)_\# \mu_0.$ In fact, $\hat \Phi_t$ is the classical flow $\hat \phi_t$ of the dressed evolution. Indeed, by the global well-posedness, the Borel set
    $$
     \mathscr G = \{ z_0 \in H^1 \oplus L^2 | \exists ! z \text{ global mild solution stating from } z_0 \} = H^1 \oplus L^2,
    $$
    and the map $\hat \Phi_t : z_0 \in \mathscr G \mapsto z(t) $ coincide with the flow $\hat \phi_t$.
\end{proof}

\subsection{Semiclassical limit of the dressing transform}

This subsection consists of linking the quantum dressing with its classical counterpart. As in the classical setting, there is nothing new from \cite[Subsection 4.7]{ammari2017sima}, but we recall here the results for clarity. The generator of the Gross transform is the self-adjoint operator
$$
T_\infty = \int \psi_x^* ( a^*(i B_\infty e^{-i k .x}) + a(i B_\infty e^{-ik.x}))\psi_x \di x.
$$
We recognize the Wick quantization of $\mathscr D_{i B_\infty}$, where we recall that the classical dressing generator is
$$
\mathscr D_{i B_\infty}(u,\alpha) = \iint (i B_\infty(k) e^{-i k . x} \overline{\alpha}(k) - i B_\infty(k) e^{i k . x} \alpha(k) )  \abs{u(x)}^2 \di x \di k.
$$
For a given parameter $\theta \in \R$ and a family of normal states $(\rho_\eps)_{\eps \in (0,1)}$, we define 
$$
\hat \rho_\eps (\theta) = e^{-\frac{i \theta}{\eps} T_\infty} \rho_\eps e^{\frac{i \theta}{\eps} T_\infty}.
$$
The main result is the following:
\begin{prop} \cite[Proposition 4.25]{ammari2017sima} \label{dressingevolution}
    Consider a family of normal states $(\rho_\eps)_{\eps \in (0,1]}$ satisfying Assumptions \eqref{assumption1} and \eqref{assumption2}. Assume that
    $$
    \mathscr M (\rho_\eps , \eps \in (0,1]) = \{ \mu_0 \} \subset \mathcal P (L^2 \oplus L^2).
    $$
    Then for all $\theta \in \R$, the family $(\hat \rho_\eps(\theta))_{\eps \in (0,1]}$ also satisfies Assumptions \eqref{assumption1} and \eqref{assumption2}, and
    $$
    \mathscr M (\hat \rho_\eps (\theta) , \eps \in (0,1]) = \{(\mathrm{D}_{i B_\infty}(\theta))_\# \mu_0 \} \subset \mathcal P (L^2 \oplus L^2),
    $$
    where $\mathrm{D}_{i B_\infty}(\theta)$ is the dressing flow as defined in Proposition \ref{gwpdressing}.
\end{prop}

\subsection{Proof of Theorem \ref{thmdynamics}}

\begin{proof}[Proof of Theorem \ref{thmdynamics}]
    Remind that
    $$
    e^{- \frac{i t}{\eps} H_\infty} \rho_\eps e^{\frac{i t}{\eps} H_\infty} = U_\infty(1)^* e^{- \frac{i t}{\eps} \hat H_\infty} U_\infty(1)  \rho_\eps U_\infty(1)^* e^{\frac{i t}{\eps} \hat H_\infty} U_\infty (1).
    $$
    By Proposition \ref{dressingevolution}, we have
    $$
    \mathscr M (U_\infty(1) \rho_\eps U_\infty(1)^*,\eps \in (0,1]) = \{ (\mathrm{D}_{i B_\infty}(1))_\# \mu_0 \}.
    $$
    Now, by Proposition \ref{dressingevolution}, the family $(U_\infty(1) \rho_\eps U_\infty(1)^*)_{\eps \in (0,1]}$ satisfies Assumptions \eqref{assumption1} and \eqref{assumption2}. Therefore by Proposition \ref{propdresseddynamics},
    $$
    \mathscr M (e^{- \frac{i t}{\eps} \hat H_\infty} U_\infty(1)  \rho_\eps U_\infty(1)^* e^{\frac{i t}{\eps} \hat H_\infty} , \eps \in (0,1]) = \{(\hat \phi_t)_\# \mathrm{D}_{i B_\infty}(1)_\# \mu_0 \} \subset \mathcal P (L^2 \oplus L^2).
    $$
    Again, by Proposition \ref{propagation}, $(e^{- \frac{i t}{\eps} \hat H_\infty} U_\infty(1)  \rho_\eps U_\infty(1)^* e^{ \frac{i t}{\eps} \hat H_\infty})_{\eps \in (0,1]}$ satisfies \eqref{assumption1} and \eqref{assumption2}, and therefore by Proposition \ref{dressingevolution},
    $$
    \mathscr M (e^{- \frac{i t}{\eps} H_\infty} \rho_\eps e^{\frac{i t}{\eps} H_\infty} , \eps \in (0,1]) = \{ \mathrm{D}_{i B_\infty}(-1)_\# (\hat \phi_t)_\# \mathrm{D}_{i B_\infty}(1)_\# \mu_0  \}.
    $$
    By Proposition \ref{linkflows}, we know that 
    $$
    (\mathrm{D}_{i B_\infty}(-1))_\# (\hat \phi_t)_\# (\mathrm{D}_{i B_\infty}(1))_\# \mu_0 = (\mathrm{D}_{i B_\infty}(-1)
    \circ \hat \phi_t \circ \mathrm{D}_{i B_\infty}(1))_\# \mu_0  = (\phi_t)_\# \mu_0,
    $$
   
    which concludes.
\end{proof}

\appendix

\section{Proof of Lemma \ref{WxiQH0}, \ref{lemmeb0} and estimates on \texorpdfstring{$\mathscr B_0(\xi)$ and $b_0(\xi)$}{Lg}} \label{appendixlemma}
\begin{proof}[Proof of Lemma \ref{WxiQH0}]
    We use a density argument and consider $\xi = (\xi_1,\xi_2) \in H^2 \oplus L^2$. For $\varphi \in Q(H_0)$, we compute
    \begin{align*}
        \norm{H_0^{1/2} W(\xi) \varphi}^2 & = \langle W(\xi) \varphi , H_0 W(\xi) \varphi \rangle \\
        &= \langle \varphi, W(\xi)^* ( \norm{\nabla u}_{L^2}^2 + \norm{\alpha}_{L^2}^2 )^{\mathrm{Wick}} W(\xi) \varphi \rangle.
    \end{align*}
    Now, by \cite[Proposition 2.10]{ammari2008ahp}, we have
    \begin{align*}
 W(\xi)^* ( \norm{\nabla u}_{L^2}^2 + \norm{\alpha}_{L^2}^2 )^{\mathrm{Wick}} W(\xi) &=  \left( \norm{\nabla u + \frac{i \eps}{\sqrt{2}} \nabla \xi_1}_{L^2}^2 + \norm{\alpha + \frac{i \eps}{\sqrt{2}} \xi_2}_{L^2}^2 \right)^{\mathrm{Wick}} \\
 &= \Bigg(  \norm{\nabla u}_{L^2}^2 + \norm{\alpha}_{L^2}^2 + \sqrt{2} \mathrm{Re} \langle u , i \eps (- \Delta \xi_1) \rangle  \\
 &\ \ \ \ \ \ + \sqrt{2} \mathrm{Re} \langle \alpha, i \eps \xi_2 \rangle + \frac{\eps^2}{2} (\norm{\nabla \xi_1}_{L^2}^2 + \norm{\xi_2}^2_{L^2})
 \Bigg)^{\mathrm{Wick}} \\
 &= H_0 + \frac{1}{\sqrt{2}} (\psi( i \eps (- \Delta \xi_1)) +\psi^*( i \eps (- \Delta \xi_1)))  \\
 &+ \frac{1}{\sqrt{2}} (a( i \eps \xi_2) + a^* ( i \eps \xi_2)) + \frac{\eps^2}{2}(\norm{\nabla \xi_1}_{L^2}^2 + \norm{\xi_2}^2_{L^2}),
    \end{align*}
which gives 
\begin{align*}
        \norm{H_0^{1/2} W(\xi) \varphi}^2
        & = \langle \varphi , \left(H_0 + \frac{\eps^2}{2} \norm{\xi}^2_{H^1 \oplus L^2} \right) \varphi \rangle + \sqrt{2} \mathrm{Re} \langle \varphi, \psi(i \eps (-\Delta \xi_1)) \varphi \rangle + \sqrt{2} \mathrm{Re} \langle \varphi, a(i \eps \xi_2) \varphi \rangle \\
        &\leq  \langle \varphi , \left(H_0 + \frac{\eps^2}{2} \norm{\xi}^2_{H^1 \oplus L^2} \right) \varphi \rangle + \sqrt{2} \eps \norm{\varphi} \norm{ (-\Delta)^{1/2} \xi_1}_{L^2} \norm{ \di \Gamma_1 (- \Delta)^{1/2} \varphi}\\
        & \ \ \ \ \ + \sqrt{2} \eps \norm{\varphi} \norm{\xi_2}_{L^2} \norm{N_2^{1/2} \varphi} \\
        & \leq   \langle \varphi , \left(H_0 + \frac{\eps^2}{2} \norm{\xi}^2_{H^1 \oplus L^2} \right) \varphi \rangle + \sqrt{2} \eps \norm{\xi}_{H^1 \oplus L^2} ( 2 \norm{\varphi}^2 + \langle \varphi, H_0 \varphi \rangle  ) \\
        &\leq (1+2\eps \norm{\xi}_{H^1 \oplus L^2})^2 \langle \varphi, (H_0+1) \varphi \rangle \\
        &= (1+2\eps \norm{\xi}_{H^1 \oplus L^2})^2 \norm{(H_0+1)^{1/2} \varphi }^2,
    \end{align*}
    where from the second to third line we have used the following inequality (see the book \cite[Theorem 5.16]{araiasao}): if $A$ is nonnegative self-adjoint one-to-one operator, $\xi \in D(A^{-1/2})$ and $\varphi \in D(\di \Gamma_1(A)^{1/2})$,
    $$
    \norm {\psi(\xi) \varphi} \leq \norm{A^{-1/2}\xi} \norm{\di \Gamma_1(A)^{1/2} \varphi }.
    $$
    Since all the terms make sense when $\xi \in H^1 \oplus L^2$, we conclude by density of $H^2 \oplus L^2$ into $H^1 \oplus L^2$. The second inequality is proven in the same fashion. For the last one, we first express
    \begin{align*}
    W(\xi) - W(\eta) &= W(\xi) (1- W(-\xi)W(\eta)) \\
    &= W(\xi) ( 1- W(-\xi+ \eta) e^{- \frac{i \eps}{2} \mathrm{Im} \langle -\xi,\eta \rangle }) \\
    &= W(\xi)(1-W(\eta - \xi)) + W(\xi) W(\eta-\xi)(1-e^{\frac{i \eps}{2} \mathrm{Im} \langle \xi,\eta \rangle }).
    \end{align*}
    By functional calculus,
    $$
    1-W(\eta-\xi) = 1- e^{i \phi (\eta-\xi)} =A(\eta-\xi) \phi (\eta-\xi),
    $$
    with $A(\eta-\xi) =  \frac{1- e^{i \phi (\eta-\xi)}}{\phi (\eta-\xi)} $ and $ \norm{A(\eta-\xi)} \leq 1$. Furthermore,
    $$
    \abs{1-e^{\frac{i \eps}{2} \mathrm{Im} \langle \xi,\eta \rangle }} \leq \frac{\eps}{2} \abs{\mathrm{Im} \langle \xi , \eta \rangle } = \frac{\eps}{2} \abs{\mathrm{Im} \langle \xi -\eta, \eta \rangle } \leq \norm{\xi-\eta}_{L^2 \oplus L^2} \norm{\eta}_{L^2 \oplus L^2}.
    $$
    Therefore,
    \begin{align*}
        \norm{(W(\xi) - W(\eta)) (N_1+N_2+1)^{-1/2}}
        & \leq \norm{W(\xi)} \norm{A(\eta-\xi)} \norm{\phi(\eta-\xi) (N_1+N_2+1)^{-1/2} } \\&+ \norm{W(\xi)} \norm{W(\eta-\xi)} \norm{\xi-\eta}_{L^2 \oplus L^2} (\norm{\xi}_{L^2 \oplus L^2}+\norm{\eta}_{L^2 \oplus L^2}) \\
        &\leq C \norm{\xi-\eta}_{L^2 \oplus L^2} (1+\norm{\xi}_{L^2 \oplus L^2}+\norm{\eta}_{L^2 \oplus L^2}).
    \end{align*}
\end{proof}

Before the proof of Lemma \ref{lemmeb0}, let us recall a standard but useful result:

\begin{lem} \label{compactness}
    Let $f,g \in L^\infty(\R^d)$ such that $f,g \underset{\pm \infty}{\longrightarrow} 0$. Then $f(x)g(D_x)$ and $g(D_x)f(x)$ are compact operators on $L^2$.
\end{lem}
\begin{proof}
    We will only prove this result for $f(x)g(D_x)$, the proof for $g(D_x)f(x)$ being handled in the same way. First we consider the case $f,g \in L^2(\R^d)$. In that case $f(x)g(D_x)$ and $g(D_x)f(x)$ are Hilbert-Schmidt operators and the corresponding norm can be explicitely computed. Indeed, let $(e_n)_{n \in \N}$ be an orthonormal basis of $L^2(\R^d)$. Then 
    \begin{align*}
        f(x) g(D_x) e_n(x) &= (2 \pi)^{-d/2} f(x)  \int_{\R^d} g(k) e^{i x .k} \mathcal F( e_n)(k) \di k \\
        &= (2 \pi)^{-d/2} f(x) \langle \overline{ \mathcal F( g e^{i x . \cdot} )}, e_n \rangle_{L^2}.
    \end{align*}
    Therefore,
    \begin{align*}
        \norm{f(x) g(D_x)}_{\mathfrak S^2(L^2(\R^d))}^2 &= \sum_{n \in \N} \norm{f(x) g(D_x) e_n}_{L^2(\R^d)}^2 \\
        &= \sum_{n \in \N} ( 2 \pi)^{-d} \int_{\R^d} \abs{f(x)}^2 \abs{\langle \overline{ \mathcal F( g e^{i x . \cdot} )}, e_n \rangle_{L^2}}^2 \di x \\
        &= (2 \pi)^{-d} \int_{\R^d} \abs{f(x)}^2 \sum_{n \in \N} \abs{\langle \overline{ \mathcal F( g e^{i x . \cdot} )}, e_n \rangle_{L^2} }^2 \di x \\
        &= (2 \pi)^{-d} \int_{\R^d} \abs{f(x)}^2  \norm{\mathcal F( g e^{i x . \cdot} )}^2_{L^2(\R^d)} \di x \\
        &= (2 \pi)^{-d} \norm{f}_{L^2(\R^d)}^2 \norm{g}_{L^2(\R^d)}^2.
    \end{align*}
    We now assume that $f,g \in L^\infty(\R^d)$ and $f,g \underset{\pm \infty}{\longrightarrow} 0$. Let $(f_m)_{m \in \N}, (g_m)_{m \in \N} \subset \mathscr C_c^\infty(\R^d)$ converging respectively to $f,g$ in $L^2(\R^d)$. Then,
    \begin{align*}
        \norm{f(x) g(D_x) - f_m(x) g_m(D_x)}&\leq  \norm{f(x) g(D_x) - f_m(x) g_m(D_x)}_{\mathfrak S^2} \\&\leq \norm{(f(x)-f_m(x))g(D_x)}_{\mathfrak S^2} + \norm{f_m(x)(g(D_x)-g_m(D_x))}_{\mathfrak S^2} \\
        &\leq (2 \pi)^{-d/2} (\norm{f-f_m}_{L^2(\R^d)} \norm{g}_{L^2(\R^d)} +\norm{f_m}_{L^2(\R^d)} \norm{g-g_m}_{L^2(\R^d)})\\
        & \underset{m \rightarrow \infty} \longrightarrow 0.
    \end{align*}
    Hence $f(x) g(D_x)$ is the limit in the operator norm of the $f_m(x) g_m(D_x)$, which are compact operators because Hilbert-Schmidt, therefore $f(x) g(D_x)$ is compact.
\end{proof}

\begin{proof}[Proof of Lemma \ref{lemmeb0}] Let $\xi = (\xi_1,\xi_2) \in H^1 \oplus L^2$. Recall that 
$$
\mathscr B_0 (\xi) = \left( \frac{i}{\sqrt{2}} \di \hat h_{\infty,I}(\cdot)(i \xi) \right)^{\mathrm{Wick}}.
$$
    We expand $\mathscr B_0 (\xi)$:
    $$
    \mathscr B_0 (\xi) = \sum_{j=1}^5 \mathscr B_{0}^j (\xi),
    $$
    with
    \begin{align*}
        \mathscr B_{0}^1 (\xi)  &= -\frac{1}{\sqrt{2}} \int \psi_x^* [a^*(f_{\sigma_0} e^{-i k .x})+a(f_{\sigma_0} e^{-i k .x})  +a^*(k B_\infty e^{-ik . x})^2 +a(k B_\infty e^{-ik . x})^2      \\
        &  + 2a^*(k B_\infty e^{-ik . x})a(k B_\infty e^{-ik . x}) -2 D_x . a(k B_\infty e^{-ik. x}) -2 a^*(k B_\infty e^{-ik. x}).D_x ] \xi_1 (x) \di x, \\
        \mathscr B_{0}^2 (\xi) &= \frac{1}{\sqrt{2}} \int \overline{\xi_1 (x)} [a^*(f_{\sigma_0} e^{-i k .x})+a(f_{\sigma_0} e^{-i k .x})  +a^*(k B_\infty e^{-ik . x})^2 +a(k B_\infty e^{-ik . x})^2      \\
        & \ \ \ \ \ \  + 2a^*(k B_\infty e^{-ik . x})a(k B_\infty e^{-ik . x}) -2 D_x . a(k B_\infty e^{-ik. x}) -2 a^*(k B_\infty e^{-ik. x}).D_x ] \psi_x \di x,  \\
        \mathscr B_{0}^3 (\xi) &= \frac{1}{\sqrt{2}} \int \psi_x^* [ \langle \xi_2, f_{\sigma_0} e^{-i k . x} \rangle- \langle f_{\sigma_0} e^{-i k .x} , \xi_2 \rangle  \\
        & - 2 D_x \langle k B_\infty e^{-i k . x  }, \xi_2 \rangle + 2 \langle \xi_2, k B_\infty e^{-ik. x} \rangle D_x] \psi_x \di x, \\
        \mathscr B_{0}^4 (\xi) &= \frac{1}{\sqrt{2}} \int \psi_x^* [- 2 a^* (k B_\infty e^{-i k . x}) \langle \xi_2, k B_\infty e^{-i k . x}  \rangle + 2 a(k B_\infty e^{-i k . x}) \langle  k B_\infty e^{-i k . x} , \xi_2 \rangle \\ 
        & +2 a^*(k B_\infty e^{-i k . x}) \langle k B_\infty e^{-i k . x}, \xi_2 \rangle - 2 \langle \xi_2, k B_\infty e^{-i k . x} \rangle a(k B_\infty e^{-i k . x}) ] \psi_x \di x, \\
        \mathscr B_{0}^5 (\xi) &= \frac{1}{ \sqrt{2}} \iint V_\infty(x-y) [\overline{\xi_1(x)} \psi_y^* \psi_x \psi_y - \psi_x^* \psi_y^* \xi_1(x) \psi_y  ] \di x \di y.
    \end{align*}
We will deal with each $\mathscr B_{0}^j (\xi)$ independently. Let $\chi \in \mathscr C_c^\infty(\R^d, [0,1])$, being identically $1$ in a neighborhood of the origin, and $\chi_m = \chi( \cdot /m)$ for every $m \in \N$. \\ \\
\textbf{Approximation of $\mathscr B_{0}^1 (\xi)$ and $\mathscr B_{0}^2 (\xi)$.} \\
We will only deal with $\mathscr B_{0}^1 (\xi)$, the second term being dealt in the same manner. We define an operator $\mathscr B_{0}^{1,m} (\xi)$ as the Wick quantization of a symbol, denoted $b_0^{1}(\xi)$ having the following expression: 
 \begin{align*}
        b_{0}^{1,m} (\xi)(u,\alpha)  &= -\frac{1}{\sqrt{2}} \int \overline{u(x)} [\langle \alpha, f_{\sigma_0} e^{-i k .x} \rangle +\langle f_{\sigma_0} e^{-i k .x}, \alpha \rangle  + \langle \alpha, k B_\infty e^{-ik . x} \rangle^2 + \langle k B_\infty e^{-ik . x},\alpha \rangle^2      \\
        & \ \ \ \ \ \  + 2\langle \alpha, k B_\infty e^{-ik . x}\rangle \langle k B_\infty e^{-ik . x}, \alpha \rangle \\
        & \ \ \ \ \ \  -2 \chi_m (D_x) D_x. \langle k B_\infty e^{-ik. x}, \alpha \rangle -2 \langle \alpha ,k B_\infty e^{-ik. x}\rangle. \chi_m (D_x) D_x ] \xi_1 (x) \di x.
\end{align*}
We first check that $b_{0}^{1,m} (\xi)$ is indeed a compact symbol. The first two terms can be rewritten under the form
$$
 \langle (u,\alpha),  \tilde b_0^{1,m} (u,\alpha) \rangle,
$$
where $\tilde b : L^2 \oplus L^2 \rightarrow L^2 \oplus L^2$ is defined as
$$
\tilde b_0^{1,m}(u,\alpha) = \left(- \frac{1}{\sqrt{2}} \xi_1(x)\left( \langle \alpha, f_{\sigma_0} e^{-i k .x} \rangle +\langle f_{\sigma_0} e^{-i k .x}, \alpha \rangle  \right) , 0 \right).
$$
Since $f_{\sigma_0}(k) \xi_1 (x) \in L^2(\R^d \times \R^d)$, $\tilde b_0^{1,m}$ is a finite rank operator, hence compact. The next three terms are dealt in the same way, writing them in the form $\langle (u,\alpha)^{\otimes_s 2} , \tilde b (u,\alpha) \rangle $ or $\langle (u,\alpha) , \tilde b (u,\alpha)^{\otimes_s 2} \rangle$. It remains to deal with the last two terms. Without loss of generality, we may only deal with the expression 
$$
\int \overline{u(x)} \chi_m (D_x) D_x . \langle k B_\infty e^{-i k . x}, \alpha \rangle  \xi_1(x) \di x = \langle (u, \alpha), ( \chi_m (D_x) D_x . \langle k B_\infty e^{-i k . x}, \alpha \rangle  \xi_1(x) ,0)).
$$
The operator 
$$
(u,\alpha) \in L^2 \oplus L^2 \mapsto( \chi_m (D_x) D_x . \langle k B_\infty e^{-i k . x}, \alpha \rangle  \xi_1(x) ,0) \in L^2 \oplus L^2
$$
is compact, by Lemma \ref{compactness}. Therefore $b_0^{1,m}$ is a compact symbol. It remains to see if $\mathscr B_0^{1,m}(\xi) := (b_0^{1,m}(\xi))^{\mathrm{Wick}}$ approximates well $\mathscr B_0^1 (\xi)$. Indeed, for every $\varphi \in D(N_1)$,
\begin{align*}
    &\abs{ \langle   (H_0+1)^{-1/2} \varphi ,  (b_0^1(\xi)-b_0^{1,m}(\xi))^{\mathrm{Wick}} (H_0+1)^{-1/2} \varphi \rangle  } \\ 
    &\lesssim   \abs{ \langle   (H_0+1)^{-1/2} \varphi ,  \int \psi_x^* (1-\chi_m(D_x)) D_x . a(k B_\infty e^{-i k . x})  \xi_1(x) \di x (H_0+1)^{-1/2} \varphi \rangle  } \\
    &+ \abs{ \langle   (H_0+1)^{-1/2} \varphi ,  \int \psi_x^* a^*(k B_\infty e^{-i k . x})  .  (1-\chi_m(D_x)) D_x   \xi_1(x) \di x (H_0+1)^{-1/2} \varphi \rangle} \\
    & \lesssim \norm{(1-\chi_m (D_x))D_x (1-\Delta_x)^{-1}} \norm{k B_\infty}_{L^2} \norm{\xi_1}_{L^2} \norm{(N_1+1)\varphi}^2,
\end{align*}
which concludes since
$$
 \norm{(1-\chi_m (D_x))D_x (1-\Delta_x)^{-1}}\underset{m \rightarrow \infty} \longrightarrow 0.
$$
\textbf{Approximation of $\mathscr B_{0}^3 (\xi)$.} \\ 
We define the symbol $b_0^{3,m}(\xi)$ as
\begin{align*}
    b_0^{3,m} (\xi) &:= \frac{1}{\sqrt{2}} \int \overline{u(x)} \chi_m(D_x) [  \langle \xi_2, f_{\sigma_0} e^{-i k . x} \rangle - \langle f_{\sigma_0} e^{-i k .x} , \xi_2 \rangle \\
    & \ \ \ \ \ \ \   - 2 D_x \langle k B_\infty e^{-i k . x  }, \xi_2 \rangle + 2 \langle \xi_2, k B_\infty e^{-ik. x} \rangle ] \chi_m(D_x) u(x) \di x
\end{align*}
We can express the latter as $b_0^{3,m}(\xi)(u,\alpha) = \langle (u,\alpha), \tilde b_0^{3,m} (u,\alpha)\rangle$, with
\begin{align*}
\tilde b_0^{3,m}(\xi)(u,\alpha) &= \bigg( \frac{1}{\sqrt{2}} \chi_m(D_x) [  \langle \xi_2, f_{\sigma_0} e^{-i k . x} \rangle - \langle f_{\sigma_0} e^{-i k .x} , \xi_2 \rangle \\
    & \ \ \ \ \ \ \   - 2 D_x \langle k B_\infty e^{-i k . x  }, \xi_2 \rangle + 2 \langle \xi_2, k B_\infty e^{-ik. x} \rangle D_x] \chi_m(D_x) u(x), 0 \bigg).
\end{align*}
Since $\xi_2, kB_\infty \in L^2$, the maps $x \mapsto  \langle \xi_2, f_{\sigma_0} e^{-i k . x} \rangle$, $x \mapsto  \langle \xi_2, k B_\infty e^{-i k . x} \rangle$ and their conjugate belong $L^\infty$ and go to zero at infinity. The same holds for $\chi_m$.
Therefore by Lemma \ref{compactness}, $\tilde b_0^{3,m}(\xi)$ is compact. Now, for every $\varphi \in D(N_1)$,
\begin{align*}
    &\abs{ \langle   (H_0+1)^{-1/2} \varphi ,  (b_0^3(\xi)-b_0^{3,m}(\xi))^{\mathrm{Wick}} (H_0+1)^{-1/2} \varphi \rangle  } \\ 
    &\lesssim \bigg| \langle   (H_0+1)^{-1/2} \varphi ,  \int \psi_x^* (1-\chi_m(D_x)) [  \langle \xi_2, f_{\sigma_0} e^{-i k . x} \rangle - \langle f_{\sigma_0} e^{-i k .x} , \xi_2 \rangle \\
    & \ \ \ \ \ \ \   - 2 D_x \langle k B_\infty e^{-i k . x  }, \xi_2 \rangle - 2 \langle \xi_2, k B_\infty e^{-ik. x} \rangle D_x]  \chi_m(D_x) \psi_x \di x (H_0+1)^{-1/2} \varphi \rangle  \bigg| \\
    &+  \bigg| \langle   (H_0+1)^{-1/2} \varphi ,  \int \psi_x^* \chi_m(D_x) [  \langle \xi_2, f_{\sigma_0} e^{-i k . x} \rangle - \langle f_{\sigma_0} e^{-i k .x} , \xi_2 \rangle \\
    & \ \ \ \ \ \ \   - 2 D_x \langle k B_\infty e^{-i k . x  }, \xi_2 \rangle - 2 \langle \xi_2, k B_\infty e^{-ik. x} \rangle D_x]  (1-\chi_m(D_x)) \psi_x \di x (H_0+1)^{-1/2} \varphi \rangle \bigg| \\
    &\lesssim \left(\norm{(1-\chi_m(D_x))(1-\Delta_x)^{-1}}_{\mathscr L(L^2)}+
    \norm{(1-\chi_m(D_x)) D_x(1-\Delta_x)^{-1}} \right) \\
    &\times \norm{\xi_2}_{L^2} \abs{ \langle   (H_0+1)^{-1/2} \varphi , \di \Gamma_1(1-\Delta_x) (H_0+1)^{-1/2} \varphi \rangle} \\
    &\lesssim \left(\norm{(1-\chi_m(D_x))(1-\Delta_x)^{-1}}+
    \norm{(1-\chi_m(D_x)) D_x(1-\Delta_x)^{-1}}_{\mathscr L(L^2)} \right) \\
    &\times \norm{\xi_2}_{L^2} \norm{ (N_1+1)^{1/2} \varphi}, \\
\end{align*}
which concludes since
$$
\norm{(1-\chi_m(D_x))(1-\Delta_x)^{-1}}+
    \norm{(1-\chi_m(D_x)) D_x(1-\Delta_x)^{-1}} \underset{m \rightarrow \infty} \longrightarrow0.
$$
\\
\textbf{Approximation of $\mathscr B_{0}^4 (\xi)$.} \\
We define the symbol $b_0^{4,m} (\xi) $ as \begin{align*}
b_{0}^{4,m} (\xi) (u,\alpha) &:= \frac{1}{\sqrt{2}} \int \overline{u(x)} \chi_m(D_x) [- 2 \langle \alpha, k B_\infty e^{-i k . x} \rangle \langle \xi_2, k B_\infty e^{-i k . x}  \rangle \\ &+ 2 \langle k B_\infty e^{-i k . x}, \alpha \rangle \langle  k B_\infty e^{-i k . x} , \xi_2 \rangle 
         + 2 \langle \alpha, k B_\infty e^{-i k . x} \rangle \langle k B_\infty e^{-i k . x}, \xi_2 \rangle  \\ &- 2 \langle \xi_2, k B_\infty e^{-i k . x} \rangle \langle k B_\infty e^{-i k . x}, \alpha \rangle ] u(x) \di x.
\end{align*}
We can express the latter as $b_{0}^{4,m} (\xi) (u,\alpha) = \langle (u,\alpha) , \tilde b_{0}^{4,m} (u,\alpha) \rangle,$ with
\begin{align*}
\tilde b_{0}^{4,m} (\xi) (u,\alpha) &= \bigg( \frac{1}{\sqrt{2}} \chi_m(D_x)[- 2 \langle \alpha, k B_\infty e^{-i k . x} \rangle \langle \xi_2, k B_\infty e^{-i k . x}  \rangle + 2 \langle k B_\infty e^{-i k . x}, \alpha \rangle \langle  k B_\infty e^{-i k . x} , \xi_2 \rangle \\ 
        & + 2 \langle \alpha, k B_\infty e^{-i k . x} \rangle \langle k B_\infty e^{-i k . x}, \xi_2 \rangle - 2 \langle \xi_2, k B_\infty e^{-i k . x} \rangle \langle k B_\infty e^{-i k . x}, \alpha \rangle ] u, 0 \bigg).
\end{align*}
Since $\xi_2, kB_\infty \in L^2$, the map $x \mapsto  \langle \xi_2, k B_\infty e^{-i k . x} \rangle$ and its conjugate are in $L^\infty$ and go to zero at infinity. The same holds for $\chi_m$. Therefore, by Lemma \ref{compactness}, the operators
$$
\chi_m(D_x)  \langle \xi_2, k B_\infty e^{-i k . x} \rangle \ \text{ and } \ \chi_m(D_x)  \langle k B_\infty e^{-i k . x},\xi_2 \rangle
$$
are compact operators. We deduce that $\tilde b_{0}^{4,m} (\xi) : L^2 \oplus L^2 \rightarrow L^2 \oplus L^2$ is compact as well. Now, for $\varphi \in D(N_1),$
\begin{align*}
    &\abs{ \langle   (H_0+1)^{-1/2} \varphi ,  (b_0^4(\xi)-b_0^{4,m}(\xi))^{\mathrm{Wick}} (H_0+1)^{-1/2} \varphi \rangle  } \\ 
    &\lesssim  \bigg| \langle   (H_0+1)^{-1/2} \varphi ,  \int \psi_x^* (1-\chi_m(D_x)) [- 2 a^* (k B_\infty e^{-i k . x}) \langle \xi_2, k B_\infty e^{-i k . x}  \rangle \\& + 2 a(k B_\infty e^{-i k . x}) \langle  k B_\infty e^{-i k . x} , \xi_2 \rangle  + 2 a^*(k B_\infty e^{-i k . x}) \langle k B_\infty e^{-i k . x}, \xi_2 \rangle \\ & - 2 \langle \xi_2, k B_\infty e^{-i k . x} \rangle a(k B_\infty e^{-i k . x}) ] \psi_x \di x (H_0+1)^{-1/2} \varphi \rangle \bigg| \\
    &\lesssim \sum_{n=0}^\infty n \eps \bigg| \langle (H_0^{(n)}+1)^{-1/2} \varphi_n, (1-\chi_m(D_{x_1})) [- 2 a^* (k B_\infty e^{-i k . x_1}) \langle \xi_2, k B_\infty e^{-i k . x_1}  \rangle \\& + 2 a(k B_\infty e^{-i k . x_1}) \langle  k B_\infty e^{-i k . x_1} , \xi_2 \rangle + 2 a^*(k B_\infty e^{-i k . x_1}) \langle k B_\infty e^{-i k . x_1}, \xi_2 \rangle \\& - 2 \langle \xi_2, k B_\infty e^{-i k . x_1} \rangle a(k B_\infty e^{-i k . x_1}) ] (H_0^{(n)}+1)^{-1/2} \varphi_n \rangle \bigg| \\ & \lesssim
    \norm{(1-\Delta_{x_1})^{-1/2}(1-\chi_m(D_{x_1}))} \norm{\xi_2}_{L^2}  \\ & \ \ \ \times \sum_{n=0}^\infty n \eps \norm{(1-\Delta_{x_1})^{1/2} (H_0^{(n)}+1)^{-1/2} \varphi_n} \norm{a^\#(k B_\infty e^{- i k .x_1}) (H_0^{(n)}+1)^{-1/2} \varphi_n}\\
    &\leq \norm{(1-\Delta_{x_1})^{-1/2}(1-\chi_m(D_{x_1}))} \norm{\xi_2}_{L^2} \sum_{n=0}^\infty n \eps \norm{\varphi_n}^2\\
    &\leq  \norm{(1-\Delta_{x_1})^{-1/2}(1-\chi_m(D_{x_1}))} \norm{\xi_2}_{L^2} \norm{(N_1+1)\varphi}^2 ,
\end{align*}
which concludes, since  
$$
\norm{(1-\Delta_{x_1})^{-1/2}(1-\chi_m(D_{x_1}))} \underset{m \rightarrow \infty} \longrightarrow 0.
$$
\\ \\
\textbf{Approximation of $\mathscr B_{0}^5 (\xi)$.} \\
We define the symbol $b_0^{5} (\xi) $ as
$$b_{0}^{5} (\xi) = \frac{1}{ \sqrt{2}} \iint V_\infty(x-y) [\overline{\xi_1(x)} \overline{u(y)} u(x) u(y) - \overline{u(x)} \overline{u(y)} \xi_1(x) u(y)] \di x \di y.$$
We can express the latter as $b_0^{5,m}(\xi) (u,\alpha) = \langle (u,\alpha), \tilde b_1 (u,\alpha)^{\otimes_s 2} \rangle + \langle (u,\alpha)^{\otimes_s 2},\tilde b_2 (u,\alpha) \rangle$, with 
\begin{align*}
    \tilde b_1 (u,\alpha)^{\otimes_s 2} &= \left( \frac{1}{\sqrt{2}} \int \overline{\xi_1(x)} V_\infty(x-y) u(x) \di x u(y)  ,0 \right),\\
    \tilde b_2 (u,\alpha) &= \left( -V_\infty(x-y) (\xi_1 \otimes_s  u)(x,y), 0 \right).
\end{align*}
Define
\begin{align*}
    \tilde b_1^m (u,\alpha)^{\otimes_s 2} &= \left( \frac{1}{\sqrt{2}} \int \overline{\xi_1(x)} \chi_m(D_y) V_\infty(x-y) u(x) \di x  \chi_m(D_y) u(y)  ,0 \right),\\
    \tilde b_2^m (u,\alpha) &= \left( - \chi_m(D_y) V_\infty(x-y) \chi_m (D_y) (\xi_1 \otimes_s  u)(x,y), 0 \right).
\end{align*}
Since $\chi_m(D_y) V_\infty(x-y) \chi_m (D_y)$ is compact and uniformly bounded with respect to $x \in \R^d$, $\tilde b_1^m$ and $\tilde b_2^m$ define compact operators. Now, for $\varphi \in D(N_1) \cap D((H_0+1)^{1/2}),$
\begin{align*}
    &\abs{ \langle   \varphi ,  (b_0^5(\xi)-b_0^{5,m}(\xi))^{\mathrm{Wick}} \varphi \rangle  } \\ 
    &\lesssim  \abs{ \langle   \varphi ,  \iint \psi_y^*  \overline{\xi_1(x)} [V_\infty(x-y)-\chi_m(D_y) V_\infty(x-y) \chi_m(D_y)]  \psi_x \psi_y \di x \di y \varphi \rangle} \\
    &+     \abs{ \langle    \varphi ,  \iint \psi_y^*  \psi_x^* [V_\infty(x-y)-\chi_m(D_y) V_\infty(x-y) \chi_m(D_y)]  \xi_1(x) \psi_y \di x \di y \varphi \rangle}\\
    &\lesssim \sum_{n=0}^\infty n \eps \abs{ \langle \varphi_n , \int \overline{\xi_1(x)}(1-\chi_m(D_{y_1})) V_\infty(x-y_1) \chi_m(D_{y_1})  \psi_x  \di x \varphi_n  \rangle} \\
    &+ \sum_{n=0}^\infty n \eps \abs{ \langle  \varphi_n , \int \overline{\xi_1(x)} V_\infty(x-y_1) (1-\chi_m(D_{y_1}))  \psi_x  \di x \varphi_n  \rangle} 
    \\
    &+ \sum_{n=0}^\infty n \eps \abs{ \langle  \varphi_n , \int \psi_x^* (1-\chi_m(D_{y_1})) V_\infty(x-y_1) \chi_m(D_{y_1})  \xi_1(x)  \di x \varphi_n  \rangle} \\
    &+ \sum_{n=0}^\infty n \eps \abs{ \langle  \varphi_n , \int \psi_x^* V_\infty(x-y_1) (1-\chi_m(D_{y_1})) \xi_1(x)  \di x \varphi_n  \rangle} 
    \\
    &\lesssim  \sum_{n=0}^\infty n \eps \norm{(H_0^{(n)}+1)^{1/2}\varphi_n}^2 \norm{(H_0^{(n)}+1)^{-1/2} \int \overline{\xi_1(x)} V_\infty(x-y_1) \psi_x \di x } \\ & \ \ \ \ \ \ \ \ \ \ \ \ \times \norm{(1-\chi_m(D_{y_1}))(1-\Delta_{y_1})^{-1/2}} 
    \\
    &\lesssim \sum_{n=0}^\infty n \eps \norm{(H_0^{(n)}+1)^{1/2}\varphi_n}^2 \norm{(1-\Delta_{y_1})^{-1/2}(1-\chi_m(D_{y_1}))} \norm{\psi(\xi_1  V_\infty(\cdot-y_1)) (H_0^{(n)}+1)^{-1/2} }  \\
    &\lesssim \sum_{n=0}^\infty n \eps \norm{(H_0^{(n)}+1)^{1/2}\varphi_n}^2 \norm{(1-\Delta_{y_1})^{-1/2}(1-\chi_m(D_{y_1}))} \norm{\xi_1}_{H^1}  \norm{V_\infty}_{L^\infty}  \\
    &\lesssim  \norm{\xi_1}_{H^1}  \norm{(1-\Delta_{y_1})^{-1/2}(1-\chi_m(D_{y_1}))} \norm{(N_1+1) (H_0+1)^{1/2}\varphi}^2.
 \end{align*}
\end{proof}

\begin{lem} \label{dominationB0m}
We have, for every $m \in \N$, the following bound
$$
\norm{(H_0+1)^{-1/2}(N_1+1)^{-1} \mathscr B_0^m(\xi) (N_1+1)^{-1}(H_0+1)^{-1/2}}  \lesssim \norm{\xi}_{H^1 \oplus L^2}.
$$
\end{lem}

\begin{proof}
    One can prove it directly doing similar estimates as in the proof of Lemma \ref{lemmeb0}. Another way of seeing it is by applying \cite[Lemma 2.4]{ammari2008ahp}, considering we have the bound on the operator norms of the (compact) symbols
    $$
    \sum_{j=1}^5 \norm{\tilde b_0^{j,m}} \lesssim \norm{\xi}_{H^1 \oplus L^2}. 
    $$
\end{proof}

\begin{lem} \label{lemmeb0m} 
    We have the estimate on the classical symbols: for every $\xi \in H^1 \oplus L^2$ and $(u,\alpha) \in H^1 \oplus L^2$,
    $$
    \abs{b_0(\xi)(u,\alpha)-b_0^m(\xi)(u,\alpha)} \leq C_m (\norm{\xi}_{H^1 \oplus L^2})(\norm{\nabla u}^2_{L^2}+ \norm{\alpha}^2_{L^2}+1)(\norm{u}_{L^2}^4+1), 
    $$
    where 
    $$
    C_m (\norm{\xi}_{H^1 \oplus L^2}) \underset{m \rightarrow \infty} \longrightarrow 0.
    $$
\end{lem}

\begin{proof}
Proposition \ref{lemmeb0} tells us that for any $\varphi \in D( (H_0+1)^{1/2}(N_1+1))$,
$$
\abs{\langle  \varphi ,(b_0 (\xi) - b_0^m (\xi))^{\mathrm{Wick}} \varphi \rangle} \leq C_m(\norm{\xi}_{H^1 \oplus L^2}) \abs{\langle \varphi, (H_0+1)(N_1+1)^2 \varphi \rangle }.
$$
Furthermore, one can recognize $ (H_0+1)(N_1+1)^2$ as the Wick quantization of an $\eps$-dependent symbol:
\begin{align*}
 (H_0+1)(N_1+1)^2 &= \bigg(  \langle u,-\Delta u\rangle (\langle u,u \rangle^2 +(2+3 \eps) \langle u,u \rangle +(\eps+1)^2) \\
 &\ \ \ \ \ \ \ \ \ +  (\langle \alpha,\alpha \rangle +1) (\langle u, u\rangle^2 +(2+\eps)\langle u,u\rangle +1) \bigg)^{\mathrm{Wick}}.
\end{align*}
According to \cite[Proposition 2.9]{ammari2008ahp}, applying the latter to the coherent states $W \left( \frac{\sqrt{2} (u,\alpha) }{i \eps} \right)  \Omega$ yields
\begin{align*}
\abs{b_0(\xi)(u,\alpha)-b_0^m(\xi)(u,\alpha)}& \leq C_m(\norm{\xi}_{H^1 \oplus L^2}) |  \langle u,-\Delta u\rangle (\langle u,u \rangle^2 +(2+3 \eps) \langle u,u \rangle +(\eps+1)^2) \\
 &\ \ \ \ \ \ \ \ \ +  (\langle \alpha,\alpha \rangle +1) (\langle u, u\rangle^2 +(2+\eps)\langle u,u\rangle +1) | \\& \leq C_m (\norm{\xi}_{H^1 \oplus L^2})(\norm{\nabla u}^2_{L^2}+ \norm{\alpha}^2_{L^2}+1)(\norm{u}_{L^2}^4+1).
\end{align*}
\end{proof}

\section{Uniqueness results of the transport equation for measures} \label{appendixliouville}

The setting is the following: we consider  a separable Hilbert space $(\mathscr Z,\langle \cdot,\cdot \rangle_{\mathscr Z})$, and a continuous vector field $X : \R \times \mathscr Z \rightarrow \mathscr Z $ bounded on bounded sets. We are interested in the initial value problem
$$
\frac{\di z}{\di t} = X(t,z), \  z(0) = z_0 \in \mathscr Z.
$$
Since in general $\mathscr Z$ is infinite-dimensional, we need to define an appropriate set of test functions. This leads to the following definition:
\begin{defn}
    We say that $\varphi : \mathscr Z \rightarrow \C$ is a smooth cylindrical function whenever there exists $n \in \N$, $(e_1,...,e_n) \in \mathscr Z^n$ and $\psi \in \mathscr C^\infty_0 (\R \times \R^n)$ such that for all $z \in \mathscr Z$,
    $$
    \varphi (t,z) = \psi (t, \mathrm{Re}\langle z, e_1 \rangle_{\mathscr Z},...,\mathrm{Re}\langle z, e_n \rangle_{\mathscr Z}).
    $$
    We denote by $\mathscr C^\infty_{0,\text{cyl}} (\R \times \mathscr Z)$ the set of cylindrical functions.
\end{defn}

    Using the same notations as above, the gradient of a cylindrical test function $\varphi \in \mathscr C^\infty_{0,\text{cyl}} (\R \times \mathscr Z)$ can be expressed as
    $$
    \nabla_{\mathscr Z} \varphi (t,z) = \sum_{i=1}^n \partial_i \psi (t, \mathrm{Re}\langle z, e_1 \rangle_{\mathscr Z},...,\mathrm{Re}\langle z, e_n \rangle_{\mathscr Z}) e_i.
    $$

\begin{rem}
    Here, the gradient is computed using the real structure of $\mathscr Z$, i.e. in the real Hilbert space $(\mathscr Z, \mathrm{Re}\langle \cdot, \cdot \rangle_{\mathscr Z})$.
\end{rem}

\begin{prop} \label{liouvilleequivalence}
Consider a continuous vector field $X: \R \times \mathscr Z \rightarrow \mathscr Z $, bounded on bounded sets. Let $t \in \R \mapsto \mu_t \in \mathscr P (\mathscr Z)$ weakly narrowly continuous such that
$$
t \mapsto \int_{\mathscr Z} \norm{X(t,z)}_{\mathscr Z} \di \mu_t (z) \di t \in L^1_{\mathrm{loc}}(\R).
$$
Then for every interval $I$ whose interior contains $0$, there is an equivalence between:
\begin{enumerate}
    \item  $\{ \mu_t  \}_{t \in I }$ satisfies the Liouville equation: $\forall \varphi \in  \mathscr C^\infty_{0,\text{cyl}} (I \times \mathscr Z)$,
    $$
    \int_I \int_{\mathscr Z} \left(  \partial_t \varphi (t,z) + \mathrm{Re} \langle  X(t,z) ,  \nabla_{\mathscr Z} \varphi (t,z)   \rangle_{\mathscr Z}   \right) \di \mu_t (z) \di t = 0.
    $$
    \item $\{ \mu_t \}_{t \in I}$ satisfies the characteristic equation: $\forall t \in I, \forall y \in \mathscr Z$,
    $$
    \mu_t \left( e^{2 i \pi \mathrm{Re}\langle y, \cdot \rangle_{\mathscr Z}  } \right) =  \mu_0 \left( e^{2 i \pi \mathrm{Re}\langle y, \cdot \rangle_{\mathscr Z}} \right  ) + 2 i \pi \int_0^t \mu_s \left( \mathrm{Re}\langle X(s,\cdot) , y \rangle_{\mathscr Z} e^{2 i \pi \mathrm{Re}\langle y, \cdot \rangle_{\mathscr Z}} \right) \di s.
    $$
\end{enumerate}
\end{prop}

Let us now state a result obtained in \cite{afs24}. We endow $\mathscr C (\R,\mathscr Z)$ with the compact open topology, and denote $e_t : (x,\gamma) \in \mathscr Z \times \mathscr C (\R,\mathscr Z) \mapsto \gamma(t) \in \mathscr Z $.

\begin{prop} (global superposition principle) \label{globalsuperposition}
    Consider a continuous vector field $X: \R \times \mathscr Z \rightarrow \mathscr Z $, bounded on bounded sets. Let $t \in \R \mapsto \mu_t \in \mathscr P (\mathscr Z)$ be a weakly narrowly continuous solution of the Liouville equation such that
$$
t \mapsto \int_{\mathscr Z} \norm{b(t,z)}_{\mathscr Z} \di \mu_t (z) \di t \in L^1_{\mathrm{loc}}(\R).
$$
Then there exists a probability measure $\nu \in \mathscr P (\mathscr Z \times \mathscr C(\R,\mathscr Z))$ such that
\begin{itemize}
    \item $\di \nu (z,\gamma)$ a.s., $\gamma$ is an integral solution of the initial value problem such that $\gamma(0) = x$.
    \item $\mu_t = (e_t)_* \nu$ for all $t \in \R$.
\end{itemize}
\end{prop}

\bibliographystyle{plain}

\end{document}